\documentclass[a4paper,UKenglish,cleveref, autoref, thm-restate]{lipics-v2021}

\nolinenumbers

\usepackage[T1]{fontenc}
\usepackage{todonotes}
\usepackage{graphicx}
\usepackage{mathtools}
\usetikzlibrary{matrix, positioning,calc} 
\DeclarePairedDelimiter{\ceil}{\lceil}{\rceil}
\usepackage{color, colortbl}
\definecolor{mgray}{RGB}{212, 212, 212}
\definecolor{darkgray}{RGB}{150, 150, 150}

\newcolumntype{g}{>{\columncolor{darkgray!70}}c}

\usepackage{bibentry}
\usepackage{amsmath}
\usepackage{amssymb}

\newcommand{\dd }{\mathinner{.\,.}}
\newcommand{\rev}[1]{#1^\mathrm{rev}}

\newcommand{\fix}[1]{\textcolor{black}{#1}}
\newcommand{\fixx}[1]{\textcolor{black}{#1}}
\newcommand{\gonzalo}[1]{\textcolor{green!50!black}{#1}}

\newcommand{\removed}[1]{}
\newcommand{\no}[1]{}
\newtheorem{property}{Property}

\newcommand{\shorten}[1]{#1}

\title{Counting on General Run-Length Grammars}
\author{Gonzalo Navarro}{Center for Biotechnology and Bioengineering (CeBiB) \\
Department of Computer Science, University of Chile, Chile}{gnavarro@dcc.uchile.cl}{}{Funded by Basal Funds FB0001, Mideplan, Chile, and Fondecyt Grant 1-230755, Chile.}
\author{Alejandro Pacheco}{Center for Biotechnology and Bioengineering (CeBiB) \\
Department of Computer Science, University of Chile, Chile}{gnavarro@dcc.uchile.cl}{}{Funded by Basal Funds FB0001, Mideplan, Chile, Fondecyt Grant 1-230755, Chile, and ANID/Scholarship Program/DOCTORADO BECAS CHILE/2018-21180760.} 
\authorrunning{G. Navarro and A. Pacheco}
\Copyright{Gonzalo Navarro and A. Pacheco}

\EventEditors{}
\EventLongTitle{35th Annual Symposium on Combinatorial Pattern Matching (CPM 2025)}
\EventShortTitle{CPM 2025}
\EventAcronym{CPM}
\EventYear{2025}
\EventDate{}
\EventLocation{Milan, Italy}
\EventLogo{}
\SeriesVolume{}
\ArticleNo{}

\begin{document}%
\maketitle              % typeset the header of the contribution
\begin{abstract}
%Context-Free grammars (CFGs) have shown to be practical for indexing and computation over compressed data, especially if highly repetitive. Enhancing CFGs with run-length rules, which yields Run-Length Context-Free grammar (RLCFGs), has improved compression on repetitive sequences and proved essential in tools like locally consistent parsing. RLCFGs are intrinsically more powerful than CFGs, producing potentially smaller grammar sizes. While extending the results for computation on CFG-compressed data to RLCFGs has succeeded in queries like substring extraction and pattern occurrence location, counting pattern occurrences in arbitrary RLCFGs has defied those attempts, except in particular cases. This paper introduces 
We introduce a data structure for counting pattern occurrences in texts compressed with any run-length context-free grammar. Our structure uses space proportional to the grammar size and counts the occurrences of a pattern of length $m$ in a text of length $n$ in time \(O(m\log^{2+\epsilon} n)\), for any constant \(\epsilon > 0\) \fixx{chosen at indexing time}. This \removed{closes}\fixx{is the first solution to} an open problem posed by Christiansen et al.~[ACM TALG 2020] and enhances our abilities for computation over compressed data; we give an example application.
\end{abstract}

\newpage

\section{Introduction}

Context-free grammars (CFGs) have proven to be an elegant and efficient model for data compression. The idea of grammar-based compression \cite{SS82,KY00} is, given a text \(T[1\dd n]\), to construct a context-free grammar \(G\) of size $g$ that only generates \(T\). One can then store $G$ instead of \(T\), which achieves compression if $g \ll n$. Compared to more powerful compression methods like Lempel-Ziv \cite{LZ76}, grammar compression offers efficient direct access to arbitrary snippets of $T$ without the need of full decompression \cite{Ryt03,BLRSRW15}. This has been extended to offering indexed searches (i.e., in time $o(n)$) for the occurrences of string patterns in $T$ \cite{CN10,GGKNP12,DBLP:journals/jcss/ClaudeNP21,christiansen2020optimal,Navacmcs20.2}, as well as more complex computations over the compressed sequence \cite{KRRW15,GKKL18,GNPjacm19,GJL21,Navcpm23.1,KK23}.
Since finding the smallest grammar $G$ representing a given text $T$ is NP-hard \cite{Ryt03,CLLPPSS05}, many algorithms have been proposed to find small grammars for a given text \cite{LM00,Ryt03,NMWM04,Sak05,MST12,Jez15,Jez16}. Grammar compression is particularly effective when handling repetitive texts; indeed, the size $g^*$ of the smallest grammar representing $T$ is used as a measure of its repetitiveness \cite{Navacmcs20.3}.

Nishimoto et al.~\cite{NIIBT16} proposed enhancing CFGs with ``run-length rules'' \fixx{to improve the compression of} \removed{to handle irregularities  when compressing} repetitive strings. These run-length rules have the form \(A \rightarrow B^s\), where \(B\) is a terminal or a non-terminal symbol and  $s \geq 2$ is an integer. CFGs that may use run-length rules are called run-length context-free grammars (RLCFGs). Because CFGs are RLCFGs, the size $g_{rl}^*$ of the smallest RLCFG generating $T$ always satisfies $g_{rl}^* \le g^*$, and it can be $g_{rl}^* = o(g^*)$ in text families as simple as $T=a^n$, where $g_{rl}^*=O(1)$ and $g^*=\Theta(\log n)$. 

The use of run-length rules has become essential to produce grammars with size guarantees and convenient regularities that speed up indexed searches and other computations \cite{KRRW15,GKKL18,GNPjacm19,christiansen2020optimal,KK23,kociumaka2024near}. The progress made in indexing texts with CFGs has been extended to RLCFGs, reaching the same status in most cases. These functionalities include extracting substrings, computing substring summaries, and locating all the occurrences of a pattern string  \cite[App.~A]{christiansen2020optimal}. It has also been shown that RLCFGs can be balanced \cite{NOUspire22} in the same way as CFGs \cite{GJL21}, which simplifies many compressed computations on RLCFGs. 

Interestingly, {\em counting}, that is, determining how many times a pattern occurs in the text without spending the time to list those occurrences, can be done efficiently on CFGs, but not so far on RLCFGs. Counting is useful in various fields, such as pattern discovery and ranked retrieval, for example to help determine the frequency or relevance of a pattern in the texts of a collection \cite{NavACMcs14}.

Navarro \cite{navarro2019document} showed how to count the occurrences of a pattern $P[1\dd m]$ in $T[1\dd n]$ in \(O(m^2 + m\log^{2+\epsilon} n )\) time using \(O(g)\) space if a CFG of size $g$ represents $T$, for any constant $\epsilon>0$ \fixx{chosen at indexing time}. Christiansen et al. improved this time to $O(m\log^{2+\epsilon} n )$ by using more recent underlying data structures for tries. Christiansen et al.~\cite{christiansen2020optimal} and Kociumaka et al.~\cite{kociumaka2024near} \removed{managed to efficiently count on} \fixx{extended the result to} {\em particular} RLCFGs, \fixx{even achieving optimal $O(m)$ time by using additional space,} but could not extend their mechanism to general RLCFGs. \removed{Their paper \cite{christiansen2020optimal} finishes, referring to counting, with ``However, this holds only for CFGs. Run-length rules introduce significant challenges [...] An interesting open problem is to generalize this solution to arbitrary RLCFGs.''}

In this paper we \removed{close}\fixx{give the first solution to} this open problem, by introducing an index that \removed{efficiently }counts the occurrences of a pattern $P[1\dd m]$ in a text $T[1\dd n]$ represented by a RLCFG of size $g_{rl}$. Our index uses $O(g_{rl})$ space and answers queries in time \(O(m\log^{2+\epsilon} n)\) for any constant \(\epsilon > 0\) \fixx{chosen at indexing time}. This is the same time complexity that holds for CFGs, which puts \fixx{on par} our capabilities to handle RLCFGs \removed{on par with those we have to handle}\fixx{and} CFGs on all the considered functionalities. As an example of our new capabilities, we show how a recent result on finding the maximal exact matches of $P$ using CFGs \cite{navarro2023computing} can now run on RLCFGs.

While our solution builds on the ideas developed for CFGs and particular RLCFGs \cite{navarro2019document,christiansen2020optimal,kociumaka2024near}, arbitrary RLCFGs lack crucial structure that holds in those particular cases, namely that if there exists a run-length rule $A \rightarrow B^s$, then the period \cite{crochemore2002jewels} of the string represented by $A$ is the length of that of $B$. We show, however, that the general case still retains some structure relating the shortest periods of $P$ and the string represented by $A$. We exploit this relation to develop a solution that, while considerably more complex than that for those particular cases, retains the same theoretical guarantees obtained for CFGs.

\section{Basic Concepts}

\subsection{Strings}

A \textit{string} $S[1\dd n] = S[1]\cdot S[2]\cdots S[n]$ is a sequence of symbols, where each symbol belongs to a finite ordered set of integers called an \textit{alphabet} $\Sigma = \{1,2,\ldots,\sigma \}$. The {\em length} of $S$ is denoted by $|S| = n$. We denote with $\varepsilon$ the empty string, where $|\varepsilon| = 0$. A {\em substring} of $S$ is $S[i\dd j] = S[i]\cdot S[i+1]\cdots S[j]$ (which is $\varepsilon$ if $i>j$). A {\em prefix (suffix)} is a substring of the form $S[\dd j] = S[1\dd j]$ ($S[j\dd] = S[j\dd n]$); we also say that $S[\dd j]$ ($S[j\dd]$) {\em prefixes} ({\em suffixes}) $S$. We write $S \sqsubseteq S'$ if $S$ prefixes $S'$, and $S \sqsubset S'$ if in addition $S \not= S'$ ($S$ strictly prefixes $S'$).
%A \textit{text} string $T=T[1\dd n]$ is a string terminated by a special symbol $T[n] = \$ \notin \Sigma$, where $\$$ is smaller than any symbol in $\Sigma$. 

We denote with $S \cdot S'$ the {\em concatenation} of $S$ and $S'$. A {\em power} $t \in \mathbb{N}$ of a string $S$, written $S^t$, is the concatenation of $t$ copies of $S$. The {\em reverse} string of $S[1\dd n] = S[1]\cdot S[2]\cdots S[n]$ refers to \(\rev{S[1\dd n]} = S[n]\cdot S[n-1]\cdots S[1]\). %\textcolor{red}{Se usa lo del \$?}
We also use the term {\em text} to refer to a string.

\subsection{Periods of strings}

Periods of strings \cite{crochemore2002jewels} are crucial in this paper. We recall their definition(s) and a key property, the renowned Periodicity Lemma.

\begin{definition} \label{def:period}
A string \( S[1\dd n] \) has a {\em period} \( 1 \le p \le n \) if, equivalently,
\begin{enumerate} 
\item it consists of \( \lfloor n/p \rfloor \) consecutive copies of \( S[1\dd p] \) plus a (possibly empty) prefix of \( S[1\dd p] \), that is, $S = (S[1\dd p]^{\lceil n/p\rceil})[1\dd n]$; or
\item $S[1\dd n-p] = S[p+1 \dd n]$; or
\item $S[i+p]=S[i]$ for all $1\le i \le n-p$. 
\end{enumerate}
We also say that $p$ is a period of $S$. We define $p(S)$ as the shortest period of \( S \) and say \( S \) is {\em periodic} if \( p(S) \leq n/2 \). 
\end{definition}

\begin{lemma}[\cite{periodicity}] \label{lem:period}
If $p$ and $p'$ are periods of $S$ \fix{and $|S| \ge p+p'-\gcd(p,p')$}, then $\gcd(p,p')$ is a period of $S$. Thus, $p(S)$ divides all other periods \fix{$p \le |S|/2$} of $S$.
\end{lemma}

\no{
\fix{
\begin{corollary}\label{col:period_p_A}
    Let $S = R^sR'$, with $s \ge 2$ and $R' \sqsubset R$. Then $p(S)=p(RR)$. 
\end{corollary}
}
\fix{
\begin{proof}
    Since $|R|$ is a period of $RR$, $p = p(RR)$ divides $|R|$ by Lemma~\ref{lem:period}, and thus $R = R[1\dd p]^k$ for $k=|R|/p$. By branch 1 of Def.~\ref{def:period}, then, $p$ is also a period of $S = R^sR' = R[1\dd p]^{ks+q}R[1\dd r]$, where $|R'| = pq+r$ and $0 \le r < p$. Suppose that the shortest period of $S$ is some $p'=p(S) < p$. Then, by Lemma~\ref{lem:period}, $p'$ divides $p$. Therefore, $p'$ is also a period of $RR = R[1\dd p]^{2k} = S[1\dd p]^{2k}=S[1\dd p']^{2kp/p'}$, again by branch 1 of Def.~\ref{def:period}; a contradiction. Thus $p = p(S) = p(RR)$. 
\end{proof}
}}

\subsection{Karp-Rabin signatures} \label{sec:kr}

Karp--Rabin \cite{KR87} fingerprinting assigns a \removed{{\em signature}}\fixx{function} $\removed{\kappa}{k}(S) = (\sum_{i=1}^{m} S[i] \cdot c^{i-1}) \bmod \mu$ to the string $S[1 \dd m]$, where $c$ is a suitable integer and $\mu$ a prime number. Bille et al.~\cite{bille2014time} showed how to build, in $O(n\log n)$ expected time, a \fixx{{\em Karp--Rabin signature $\kappa(S)$} built from a pair of Karp--Rabin functions, which has} no collisions between substrings $S$ of $T[1\dd n]$. We always assume those kind of signatures in this paper.

A well-known property is that we can compute the signatures of all the prefixes $S[\dd j] \sqsubseteq S$ in time $O(m)$, and then obtain any $\kappa(S[i\dd j])$ in constant time by using arithmetic operations.

\subsection{Range summary queries on grids} \label{sec:rangesum}

A discrete grid of $r$ rows and $c$ columns stores points at integer coordinates $(x,y)$, with $1 \le x \le c$ and $1 \le y \le r$. Grids with $m$ points can be stored in $O(m)$ space, so that some {\em summary} queries are performed on {\em orthogonal ranges} of the grid. In particular, one can associate an integer with each point, and then, given an orthogonal range $[x_1,x_2] \times [y_1,y_2]$, compute the {\em sum} of all the integers associated with the points in that range. Chazelle \cite{Cha88} showed how to run that query in time \(O(\log^{2+\epsilon} m)\), for any constant \(\epsilon > 0\), in $O(m)$ space.

\subsection{Grammar compression and parse trees}\label{sec:grammar}

A {\em context-free grammar (CFG)} \( G = (V, \Sigma, R, S) \) is a language generation model 
consisting of a finite set of nonterminal symbols \( V \) and a finite set of terminal symbols \( \Sigma \), disjoint from \( V \). 
The set \( R \) contains a finite set of production rules \( A \rightarrow \alpha \), where \( A \) is a nonterminal symbol and \( \alpha \)
is a string of terminal and nonterminal symbols. The language generation process starts from a sequence formed by just the nonterminal \( S \in V \) and, iteratively, chooses a rule $A \rightarrow \alpha$ and replaces an occurrence of $A$ in the sequence by $\alpha$, until the sequence contains only terminals. The {\em size} of the grammar, \(g = |G|\), is the sum of the lengths of the right-hand sides of the rules, $g = \sum_{A \rightarrow \alpha \in R} |\alpha|$. Given a string \( T \), we can build a CFG \( G \) that generates only \( T \). Then, especially if $T$ is repetitive, \(G\) is a compressed representation of \(T\). 
The {\em expansion} \(\exp(A)\) of a nonterminal $A$ is the string generated by \(A\), for instance \(\exp(S) = T\); for terminals $a$ we also say $\exp(a)=a$. We use $|A| = |\exp(A)|$ and $p(A)=p(\exp(A))$.

The {\em parse tree} of a grammar is an ordinal labeled tree where the root is labeled with the initial \removed{rule}\fixx{symbol} \(S\), the leaves are labeled with terminal symbols, and internal nodes are labeled with nonterminals. If \(A \rightarrow \alpha_1 \cdots \alpha_t \), with $\alpha_i \in V \cup \Sigma$, then a node \(v\) labeled $A$ has \(t\) children labeled, left to right, \(\alpha_1, \ldots, \alpha_t\). 
A more compact version of the parse tree is the {\em grammar tree}, which is obtained by pruning the parse tree such that only one internal node labeled $A$ is kept for each nonterminal $A$, while the rest become leaves. Unlike the parse tree, the grammar tree of $G$ has only \(g + 1\) nodes. Consequently, the text \(T\) can be divided 
into at most \(g\) substrings, called {\em phrases}, each being the expansion of a grammar tree leaf. The starting phrase positions constitute a 
{\em string attractor} of the text \cite{KP18}. Therefore, all text substrings \fixx{of length more than 1} have at least one occurrence that crosses a phrase boundary.

\subsection{Run-length grammars}

{\em Run-length CFGs (RLCFGs)} \cite{NIIBT16} extend CFGs by allowing in \(R\) rules of the form \(A \rightarrow \beta^s\), where \(s \geq 2\) is an integer and \(\beta\) is a string of terminals and nonterminals. These rules are equivalent to rules \(A \rightarrow \beta \cdots \beta\) with \(s\) repetitions of \(\beta\). However, the length of the right-hand side of the rule \(A\) is defined as $|\beta|+1$, not \fixx{$s \cdot |\beta|$}. To simplify, we will only allow run-length rules of the form $A \rightarrow B^s$, where $B$ is a single terminal or nonterminal; this does not increase \removed{their}the asymptotic \fixx{grammar} size because we can rewrite $A \rightarrow B^s$ and $B \rightarrow \beta$ for a fresh $B$.

RLCFGs are never larger than general CFGs, and they can be asymptotically smaller. For example, the size $g_{rl}^*$ of the smallest RLCFG that generates $T$ is in \(O(\delta \log\frac{n\log|\Sigma|}{\delta\log n})\), where \(\delta\) is a measure of repetitiveness based on substring complexity~\cite{raskhodnikova2013sublinear,KNP23}, but such a bound does not always hold for the size $g^*$ of the smallest grammar. The maximum stretch between $g^*$ and $g_{rl}^*$ is $O(\log n)$, as we can replace each rule $A \rightarrow B^s$ by $O(\log s)$ CFG rules.

We denote the size of an RLCFG $G$ as \(g_{rl} = |G|\). To maintain the invariant that the grammar tree has $g_{rl}+1$ nodes, we represent rules $A \rightarrow B^s$ as a node labeled $A$ with two children: the first is $B$ and the second is a special leaf $B^{[s-1]}$, denoting $s-1$ repetitions of $B$.

\section{Grammar Indexing for Locating} \label{sec:grammarindex}
%In the context of text compression, grammars are dictionary compressors that replace text substrings with references to a dictionary of strings, effectively exploiting text repetitiveness and surpassing the entropy lower bound on highly repetitive texts. 
A {\em grammar index} represents a text $T[1\dd n]$ using a grammar $G$ that generates only $T$. As opposed to mere compression, the index supports three primary pattern-matching queries: {\em locate} (returning all positions of a pattern in the text), {\em count} (returning the number of times a pattern appears in the text), and {\em extract} (extracting any desired substring of $T$). In order to locate, grammar indexes identify ``initial'' pattern occurrences and then track their ``copies'' throughout the text. The former are the {\em primary occurrences}, definde as those that cross phrase boundaries, and the latter are the {\em secondary occurrences}, which are confined to a single phrase. This approach\removed{, introduced by K\"arkk\"ainen and Ukkonen} \cite{karkkainen1996lempel}\removed{,} forms the basis of most grammar indexes \cite{CN10,claude2012improved,DBLP:journals/jcss/ClaudeNP21} and related ones \cite{GGKNP12,KN13,FGHP13,GGKNP14,FKP18,BEGV18,NP18,TKNIBT20}, which first locate the primary occurrences and then derive their secondary occurrences through the grammar tree.  
  
As mentioned in Section~\ref{sec:grammar}, the grammar tree leaves cut the text into phrases. In order to report each primary occurrence of a pattern $P[1\dd m]$ exactly once, let $v$ be the lowest common ancestor of the first and last leaves the occurrence spans; $v$ is called the {\em locus node} of the occurrence. Let $v$ have $t$ children and the first leaf that covers the occurrence descend from the $i$th child of $v$. If $v$ represents $A \rightarrow \alpha_1 \cdots \alpha_t$, it follows that $\exp(\alpha_i)$ finishes with a pattern prefix $R = P[1\dd q]$ and that $\exp(\alpha_{i+1})\cdots\exp(\alpha_t)$ starts with the suffix $Q = P[q+1\dd m]$. We will denote such {\em cuts} as $P = R \mid Q$. \fixx{The alignment of $R \mid Q$ within $\exp(\alpha_i) \mid \exp(\alpha_{i+1})\cdots\exp(\alpha_t)$ is the only possible one for that primary occurrence}\removed{This is the only cut of $P$ that will find this primary occurrence}.

Following the \fixx{original scheme \cite{karkkainen1996lempel},}\removed{scheme of K\"ark\"ainen and Ukkonen, classic} grammar indexing builds two sets of strings, \(\mathcal{X}\) and \(\mathcal{Y}\), to find primary occurrences \cite{CN10,claude2012improved,DBLP:journals/jcss/ClaudeNP21}. For each grammar rule $A \rightarrow \alpha_1\cdots\alpha_t$, the set \(\mathcal{X}\) contains all the reverse expansions of the children of $A$, \(\rev{\exp(\alpha_{i})}\), and \(\mathcal{Y}\) contains all the expansions of the nonempty rule suffixes, \(\exp(\alpha_{i+1}) \cdots \exp(\alpha_t)\). Both sets are sorted lexicographically and placed on a grid with (less than) \(g\) points, $t-1$ for each rule $A \rightarrow \alpha_1\cdots\alpha_t$. Given a pattern \( P[1\dd m] \), for each cut \(P = R \mid Q\), we first find the lexicographic ranges \([s_x,e_x]\) of \(\rev{R}\) in \(\mathcal{X}\) and \([s_y,e_y]\) of \(Q\) in \(\mathcal{Y}\). Each point \((x,y) \in [s_x,e_x] \times [s_y,e_y]\) represents a primary occurrence of \(P\). Grid points are augmented with their locus node \(v\) and offset \(|\exp(\alpha_1) \cdots \exp(\alpha_i)|\).

\begin{figure}[t]
    \begin{minipage}{0.6\textwidth} 
        \scriptsize
        \resizebox{\textwidth}{!}{
        \begin{tikzpicture}
        
                    \tikzset{X/.style={fill=lightgray!20}}
                    \tikzset{L/.style={fill=darkgray!70}}
                    \tikzset{V/.style={fill=mgray}}
             
                % Nodos
        	    \node [] (38)  at  (-6    ,3) {{\small $X_4$}};
        		
                    \node [L] (33)  at  (-10.87    ,2) {\textcolor{white}{{\small $X_2$}}};
        		\node [] (34)  at  (-8        ,2) {{\small $X_3$}};
        		\node [] (26)  at  (-5.25    , 2) {{\small $X_6$}};
        		\node [] (27)  at  (-4.5     , 2) {{\small $X_1$}};
        		\node [] (37)  at  (-3        ,2) {{\small $X_3$}};
        	
        	    \node [] (17)  at  (-12      , 1) {{\small $X_1$}};
        		\node [] (18)  at  (-11.25   , 1) {{\small $X_5$}};
        		\node [] (19)  at  (-10.5    , 1) {{\small $X_8$}};
        		\node [] (20)  at  (-9.75    , 1) {{\small $X_1$}};
        		\node [] (21)  at  (-9       , 1) {{\small $X_7$}};
        		\node [] (35)  at  (-7     ,1) {{\small $X_2$}};
        	    \node [] (39)  at  (-3.75     ,1) {\textcolor{lightgray}{{\small $X_7$}}};
        		\node [] (40)  at  (-2     ,1) {\textcolor{lightgray}{{\small $X_2$}}};
          
        	    \node [] (41)  at  (-8.25      , 0) {\textcolor{lightgray}{{\small $X_1$}}};
        		\node [] (42)  at  (-7.5   , 0) {\textcolor{lightgray}{{\small $X_5$}}};
        		\node [] (43)  at  (-6.75    , 0) {\textcolor{lightgray}{{\small $X_8$}}};
        		\node [] (44)  at  (-6    , 0) {\textcolor{lightgray}{{\small $X_1$}}};

        	    \node [] (45)  at  (-3      , 0) {\textcolor{lightgray}{{\small $X_1$}}};
        		\node [] (46)  at  (-2.25   , 0) {\textcolor{lightgray}{{\small $X_5$}}};
        		\node [] (47)  at  (-1.5    , 0) {\textcolor{lightgray}{{\small $X_8$}}};
        		\node [] (48)  at  (-0.75    , 0) {\textcolor{lightgray}{{\small $X_1$}}};

        		% Letras 
        	    \node [L] (1)  at (-12     , -1) {\textcolor{white}{{\small a}}};
        		\node [L] (2)  at (-11.25  , -1) {\textcolor{white}{{\small b}}};
        		\node [L] (3)  at (-10.5   , -1) {\textcolor{white}{{\small r}}};
        		\node [L] (4)  at (-9.75   , -1) {\textcolor{white}{{\small a}}};
    	        \node [] (5)  at (-9      , -1) {{\small d}};
        		\node [V] (6)  at (-8.25   , -1) {{\small a}};
        		\node [V] (7)  at (-7.5    , -1) {{\small b}};
        		\node [V] (8)  at (-6.75   , -1) {{\small r}};
        		\node [V] (9)  at (-6      , -1) {{\small a}};
        		\node [] (10) at (-5.25   ,  -1) {{\small c}};
        		\node [] (11) at (-4.5    ,  -1) {{\small a}};
        		\node [] (12) at (-3.75   , -1) {{\small d}};
        		\node [V] (13) at (-3      , -1) {{\small a}};
        		\node [V] (14) at (-2.25   , -1) {{\small b}};
        		\node [V] (15) at (-1.5    , -1) {{\small r}};
        		\node [V] (16) at (-0.75   , -1) {{\small a}};
        	
        	   % Líneas
                    \draw [lightgray!50] (38) to (33);
        		\draw [lightgray!50] (38) to (34);
        		\draw [lightgray!50] (38) to (26);
        		\draw [lightgray!50] (38) to (27);
        		\draw [lightgray!50] (38) to (37);
        %  		--------------------------------------
        		
        		\draw [lightgray!50] (33) to (17);
        		\draw [lightgray!50] (33) to (18);
        		\draw [lightgray!50] (33) to (19);
        		\draw [lightgray!50] (33) to (20);
        		
        		\draw [lightgray!50] (34) to (21);
        		\draw [lightgray!50] (34) to (35);

                    \draw [dashed,lightgray!40] (37) to (39);
        		\draw [dashed,lightgray!40] (37) to (40);

                    \draw [dashed,lightgray!40] (35) to (41);
        		\draw [dashed,lightgray!40] (35) to (42);
        		\draw [dashed,lightgray!40] (35) to (43);
        		\draw [dashed,lightgray!40] (35) to (44);

                    \draw [dashed,lightgray!40] (40) to (45);
        		\draw [dashed,lightgray!40] (40) to (46);
        		\draw [dashed,lightgray!40] (40) to (47);
        		\draw [dashed,lightgray!40] (40) to (48);

        		\draw [lightgray!50] (17) to (1);
        		\draw [lightgray!50] (18) to (2);
        		\draw [lightgray!50] (19) to (3);
        		\draw [dashed,lightgray!50] (20) to (4);
        		\draw [lightgray!50] (21) to (5);
        		\draw [dashed,lightgray!40] (41) to (6);
        		\draw [dashed,lightgray!40] (42) to (7);
        		\draw [dashed,lightgray!40] (43) to (8);
        		\draw [dashed,lightgray!40] (44) to (9);
        		\draw [lightgray!50] (26) to (10);
	           	\draw [dashed,lightgray!50] (27) to (11);
        		\draw [dashed,lightgray!40] (39) to (12);
        		\draw [dashed,lightgray!40] (45) to (13);
        		\draw [dashed,lightgray!40] (46) to (14);
        		\draw [dashed,lightgray!40] (47) to (15);
        		\draw [dashed,lightgray!40] (48) to (16);

        %  		--------------------------------------
        		 % \draw[->, darkgray, thick] (33) to[bend left=20] (35);
           
        		 % \draw[->, dotted,thick] (33) to[bend left=30] (38);
           
        		 % \draw[->, dotted, thick] (35) to[bend right=20] (34);
           
        		 % \draw[->, darkgray, thick] (34) to[bend left=20] (37);

           %        \draw[->, dotted,thick] (34) to[bend right=20] (38);
                  
           %        \draw[->, dotted,thick] (37) to[bend right=20] (38);
           
        \end{tikzpicture}
        }
        
    \end{minipage}  
    \begin{minipage}{0.4\textwidth} 
        \tiny
        \resizebox{\textwidth}{!}{
 
        \begin{tabular}{r@{~}| c@{~} c@{~} c@{~} @{~~}g@{~~~} c@{~} c@{~~~}  c@{~}  c@{~} }
          % Encabezados rotados con \rotatebox
          & \rotatebox[origin=l]{90}{\texttt{a}} 
          & \rotatebox[origin=l]{90}{\texttt{abra}} 
          & \rotatebox[origin=l]{90}{\texttt{adabra}}
          & \rotatebox[origin=l]{90}{\textcolor{white}{\texttt{bra}}}
          & \rotatebox[origin=l]{90}{\texttt{cadabra}} 
          & \rotatebox[origin=l]{90}{\texttt{dabra}} 
          & \rotatebox[origin=l]{90}{\texttt{dabracadabra}}
          & \rotatebox[origin=l]{90}{\texttt{ra}}\\
          \hline
          \rowcolor{darkgray!70}\textcolor{white}{\texttt{a}}                   & & & & \textcolor{white}{4} & & \textcolor{white}{12} & &  \\  
          \texttt{arba}                & & & & & & & 7 &  \\  
          \texttt{arbad}               & & & & & 10 & & &  \\  
          \texttt{b}                   & & & & & & & & 5  \\  
          \texttt{c}                   & & & 11 & & & & &  \\  
          \texttt{d}                   & & 9 & & & & & & \\  
          \texttt{r}                   & 6 & & & & & & &  \\  
        \end{tabular}
         }
    \end{minipage}

\caption{On the left, a grammar tree for $T=\mathtt{abracadabra}$ (with straight solid edges), so $\exp(X_4)=T$. Dashed edges were removed from the parse tree. \fixx{The only primary occurrence of $P=\mathtt{abra}$ in $T$ is marked with dark gray on the bottom; the secondary ones are in light gray.} On the right, the grid used for searching primary occurrences. Gray stripes indicate the search ranges corresponding to the partition \( P = R \ | \ Q \), where \( R = \texttt{a} \) and \( Q = \texttt{bra} \). The value $4$ stored in the resulting cell is the preorder of the child $X_5$ of the locus node $X_2$ where $Q$ starts. \removed{On the left again, the process of discovering secondary occurrences in the grammar tree. The locus node \( X_2 \) is highlighted in dark gray, while the solid arcs represent jumps to second mentions of the node. Dotted arcs illustrate the process of ascending to the lowest ancestor that appears at least twice in the grammar (or the root). Finally, the secondary occurrences of the pattern in the text are marked in light gray.}}
  
    \label{fig:grammar-indexing}    
\end{figure}

%\textcolor{blue}{Using a binary search, it is possible to find the intervals  \([sx, ex]\), \([sy, ey]\) on sets \(\mathcal{X}\) and \(\mathcal{Y}\). Each comparison of the search must extract a prefix of a rule or a suffix of the grammar with a maximum length of \(m\). Gasieniec et al. \cite{GKPS05} demonstrated how prefixes or suffixes can be extracted from any \(exp(A)\) in real-time (\(O(1)\) per additional symbol). Later, Christiansen et al.  \cite{christiansen2020optimal} [Theorem A.2., Lemma 6.6] shows that it is possible to extend the same result of Gasieniec et al. to get real-time extractions of suffixes and prefixes from a rules in an RLCFG. Therefore, the interval searching takes \(O(m^2\log g)\) time. Additionally, Claude and Navarro improve the time complexity to \(O(m)\), introducing two Patricia Trees \cite{Mor68}, one for the reverse expansions of grammar rules \(exp(\rev{A})\) and another for the expansions of grammar suffixes, using \(O(g \log n)\) extra bits. Considering the \(m-1\) cuts of the pattern, they obtained a time complexity of \(O(m^2)\) for the searching of grid ranges.}
%\textcolor{blue}{On the other hand, recovering the primary occurrence points \(occ_p\) on the grid can be done in \(O((1+occ_p)\log^{\epsilon} g)\), using the structure presented in Section~\ref{sec:rangesum}. Finally, the process of finding the primary occurrences over the \(m-1\) cuts cost \(O (m^2 + (m + occ_p) log^e g)\) time. }

Once we identify the locus node \(v\) (with label \(A\)) of a primary occurrence, \fixx{every other mention of $A$ or its ancestors in the grammar tree, and recursively, of the ancestors of those mentions, yields a secondary occurrence of $P$. Those are efficiently tracked and reported \cite{claude2012improved,DBLP:journals/jcss/ClaudeNP21,christiansen2020optimal}. An important {\em consistency} observation for counting is that the amount of secondary occurrences triggered by each primary occurrence is fixed.} \removed{we know that every other mention of \(A\) in the grammar tree contains a (secondary) occurrence with the same offset. Additionally, let \(u\) (with label \(C\)) be the parent of any node with label \(A\). All other nodes labeled \(C\) in the grammar tree also contain secondary occurrences of the pattern (with a corrected offset). From each primary occurrence locus $v$ labeled $A$, one recursively visits the parent of \(v\) (and adjusts the offset) and all the other mentions of \(A\) in the tree. Each recursive branch reaches the root of the grammar tree, uncovering a distinct offset where $P$ occurs in $T$. Claude and Navarro \cite{claude2012improved} showed that, if every nonterminal \(A\) appears at least twice in the grammar tree, the traversal cost amortizes to constant time per secondary occurrence, thereby modifying non-complaint grammars. Christiansen et al.~\cite{christiansen2020optimal} later showed that, if it is impossible to modify the grammar, each node can store a pointer to its lowest ancestor whose label appears at least twice in the grammar tree, using that ancestor instead of the parent. Both approaches report the secondary occurrences in optimal time. Figure \ref{fig:grammar-indexing} shows an example of this technique.}\fixx{See Figure~\ref{fig:grammar-indexing}.}

The original approach \cite{claude2012improved,DBLP:journals/jcss/ClaudeNP21} spends time $O(m^2)$ to find the ranges \([s_x,e_x]\) and \([s_y,e_y]\) for the $m-1$ cuts of $P$; this was later improved to $O(m\log n)$ \cite{christiansen2020optimal}. Each primary occurrence found in the grid ranges takes time $O(\log^\epsilon g)$ using geometric data structures, whereas each secondary occurrence requires $O(1)$ time. Overall, the $occ$ occurrences of $P$ in $T$ are listed in time $O(m\log n + occ\,\log^\epsilon g)$.

To generalize this solution to RLCFGs \cite[App.~A.4]{christiansen2020optimal}, rules $A \rightarrow B^s$ are added as a point $(x,y) = (\rev{\exp(B)},\exp(B)^{s-1})$ in the grid. \removed{To see that this}\fixx{This} suffices to capture every primary occurrence\removed{, regard the rule as} \fixx{of the corresponding rule} $A \rightarrow B \cdots B$\removed{.}\fixx{:} If there are primary occurrences with the cut $P=R \mid Q$ in $B\cdots B$, then one is aligned with the first phrase boundary, $\exp(B) \removed{\cdot}\fixx{\mid} \exp(B)^{s-1}$. Precisely, there is space to place $Q$ right after the first $t = s-\lceil |Q|/|B| \rceil$ phrase boundaries. When the point $(x,y)$ is retrieved for a given cut, then, $t$ primary occurrences are declared with offsets $|B|-|R|$, $2|B|-|R|$, $\ldots$, $t|B|-|R|$ within $\exp(A)$. \fixx{The amount of secondary occurrences triggered by each such primary occurrence still depends only on $A$.}

\section{Counting with Grammars} \label{sec:prevcount}

Navarro \cite{navarro2019document} obtained the first result in counting the number of occurrences of a pattern \( P[1\dd m] \) in a text \( T[1\dd n] \) represented by a CFG of size $g$, within time \( O(m^2 + m\log^{2+\epsilon} \fixx{g} ) \), for any constant $\epsilon>0$, and using \( O(g) \) space. His method relies on the {\em consistency} \fixx{observation above, which}\removed{of the secondary occurrences triggered across the grammar tree: given the locus node of a primary occurrence, the number of secondary occurrences it triggers is independent of the occurrence offset within the node. This property} allows enhancing the grid described in Section~\ref{sec:grammarindex} with the number $c(A)$ of (primary and) secondary occurrences associated with each point.\no{The value \(c(A)\) can be understood as the count of nodes labeled \(A\) in the parse tree representation of \(T\).} At query time, for each pattern cut, one sums the number of occurrences in the corresponding grid range using the technique mentioned in Section~\ref{sec:rangesum}. The final complexity is obtained by aggregating over all \( m - 1 \) cuts of $P$ and considering the \( O(m^2) \) time required to identify all the ranges.
Christiansen et al.~\cite[Thm.~A.5]{christiansen2020optimal} later improved this time to just \fixx{$O(m\log n + m\log^{2+\epsilon} g)$}, by using more modern techniques to find the grid range of all cuts of $P$.

Christiansen et al.~\cite{christiansen2020optimal} also
presented a method to count in \(O(m + \log^{2+\epsilon} n)\) time on a {\em particular} RLCFG\removed{with ``local consistency'' properties,} of size \(g_{rl} = O(\gamma \log (n/ \gamma))\), where \(\gamma\) is the size of the smallest string attractor \cite{KP18} of $T$. They also show that by increasing the space to \(O(\gamma \log (n/ \gamma)\log^{\epsilon}n)\) one can reach the optimal counting time, \(O(m)\). 
\removed{Local consistency allows}\fixx{The grammar properties allow} reducing the number of cuts of $P$ to check to $O(\log m)$, instead of the $m-1$ cuts used on general RLCFGs.  

Christiansen et al.\ build on the same idea of enhancing the grid with the number of secondary occurrences, but the process is considerably more complex on RLCFGs, because the consistency property exploited by Navarro \cite{navarro2019document} does not hold \fixx{on run-length rules $A \rightarrow B^s$}: the number of \removed{secondary }occurrences triggered by a primary occurrence with cut $P = R \mid Q$ found \fixx{from the point $(\rev{\exp(B)},\exp(B)^{s-1})$}\removed{within a run-length rule} depends on $s$, $|B|$, and $|R|$. 
Their counting approach relies on another property that is specific of their RLCFG \cite[Lem.~7.2]{christiansen2020optimal}:
\fixx{
\begin{property} \label{prop:basic}
For every run-length rule $A \rightarrow B^s$, the shortest period of $\exp(A)$ is $|B|$.
\end{property}
}
This property facilitates the division of the counting process into two cases\removed{: one applies when the pattern is periodic and the other when it is not}.
For each run-length rule $A \rightarrow B^s$, they introduce two points, $(x,y') = (\rev{\exp(B)},\exp(B))$ and \fix{$(x,y'') = (\rev{\exp(B)},\exp(B)^2)$}, in the grid. These points are associated with the values \(c(A)\) and \( (s-2)\cdot c(A)\), respectively. The counting process is as follows: for a cut \(P = R \mid Q\), if $Q \sqsubseteq \exp(B)$, then it will be counted \(c(A) + (s-2) \cdot c(A) = (s-1)\cdot c(A) \) times, as both points will be within the search range. If $Q$ instead exceeds \(\exp(B)\), but still $Q \sqsubseteq \exp(B)^2$, then it will be counted  \((s-2) \cdot c(A)\) times, solely by point $(x,y'')$. Finally if $Q$ exceeds \(\exp(B)^2\), then $Q$ is periodic (with $p(Q)=|B|$). 

They handle that remaining case as follows. Given a cut \(P = R \mid Q\) and the period \(p = p(Q) = |B|\), where \( |Q| > 2p\), the number of primary occurrences of this cut inside rule \(A \rightarrow B^s \) is \(s - \lceil |Q|/p\rceil\) (cf.\ the end of Section~\ref{sec:grammarindex}). Let \(D\) be the set of rules \(A \rightarrow B^s\) within the grid range of the cut, and $c(A)$ the number of (primary and secondary) occurrences of $A$. Then, the number of occurrences triggered by the primary occurrences found within symbols in \(D\) for this cut is
\begin{equation*}
    \sum_{A \rightarrow B^s \in D} c(A) \cdot s - c(A) \cdot \ceil{|Q|/p}.
\end{equation*}
For each run-length rule \(A \rightarrow B^s\), they compute a Karp--Rabin signature (Section~\ref{sec:kr}) $\kappa(\exp(B))$ and store it in a perfect hash table, associated with values 
\begin{eqnarray*}
C(B,s) &~=~& \sum\{ c(A) : A \rightarrow B^{s'} , s' \geq s\}, \\
C'(B,s) &~=~& \sum\{ s' \cdot c(A) : A \rightarrow B^{s'} , s' \geq s\}. 
\end{eqnarray*}
Additionally, for each such \(B\), the authors store the set \(s(B) = \{s : A \rightarrow B^{s}\}\). 

At query time, they calculate the shortest period $p=p(P)$. For each cut  \(P = R \mid Q\), $Q$ is periodic if \(|Q| > 2p\). If so, they compute \(k = \kappa(Q[1\dd p])\), and if there is an entry $B$ associated with \(k\) in the hash table, they add to the number of occurrences found up to then
\begin{equation*}
    C'(B,s_{min}) - C(B,s_{min})\cdot \ceil{|Q|/p},
\end{equation*}
where \(s_{min} = \min\{s \in s(B), (s-1) \cdot |B| \geq |Q|\}\) is computed using exponential search over \(s(B)\) in \(O(\log m)\) time. Note that they exploit the fact that {\em the number of repetitions to subtract}, $\lceil |Q|/p \rceil$, depends only on $p=|B|$, and not on the exponent $s$ of rules $A \rightarrow B^s$.

\removed{Since the grammar is locally consistent,}\fixx{The total counting time, on a grammar of size $g_{rl}$, is $O(m\log n + m\log^{2+\epsilon} g_{rl})$. In their particular grammar,} the number of cuts to consider is \(O(\log m)\), which allows reducing the cost of computing the grid ranges to $O(m)$. The signatures of all substrings of $P$ are also computed in \(O(m)\) time, as mentioned in Section~\ref{sec:kr}. Considering the grid searches, the total cost for counting the pattern occurrences \fixx{drops to} \(O(m + \log^{2+\epsilon} g_{rl}) \subseteq O(m+\log^{2+\epsilon} n)\).

Recently, Kociumaka et al. \cite{kociumaka2024near} employed this same approach to count the occurrences of a pattern in a smaller RLCFG that uses \(O(\delta \log \frac{n \log |\Sigma|}{\delta \log n})\) space, where \(\delta \le \gamma\). \removed{Kociumaka et al.}\fixx{They} demonstrated that the RLCFG they produce satisfies \fixx{Property~\ref{prop:basic}} \cite[Lem.~7.2]{christiansen2020optimal}, which is necessary to apply the described scheme. 

\section{Our Solution} \label{sec:oursol}

We now describe a solution to count the occurrences in arbitrary RLCFGs, where the convenient \fixx{Property~\ref{prop:basic}} used in the literature may not hold. We start with a simple \removed{but useful }observation.

\begin{lemma}\label{lem:period_div}
Let $A \rightarrow B^s$ be a rule in a RLCFG. Then $p(A)$ divides $|B|$.
\end{lemma}
\begin{proof}
Clearly $|B|$ is a period of $\exp(A)$ because $\exp(A)=\exp(B)^s$. By Lemma~\ref{lem:period}, then, \fix{since $|B| \le |A|/2$,} $p(A)$ divides $|B|$.
\end{proof}

 Some parts of our solution make use of the shortest period of $\exp(A)$. We now define some related notation.

\begin{definition}\label{def:rule_bps}
Given a rule \(A \rightarrow B^s\) with \(s \ge 2\), let $p=p(A)$ (which divides $|B|$ by Lemma \removed{s~\ref{lem:period} and} \ref{lem:period_div}). The corresponding {\em transformed rule} is $A \rightarrow \hat{B}^{s'}$, where $\hat{B}$ is a new nonterminal such that $\exp(\hat{B}) = \exp(A)[1\dd p]$, and \(s' = s \cdot (|B|/p)\).
\end{definition}

\shorten{There seems to be no way to just transform all run-length rules (which would satisfy \removed{the convenient property}\fixx{Property~\ref{prop:basic},} $p(A)=|\hat{B}|$) without blowing up the RLCFG size by a logarithmic factor. We will use another approach instead.} 
%\textcolor{blue}{
We classify the rules into two categories.
\begin{definition}
Given a rule \(A \rightarrow B^s\) with \(s \ge 2\), we say that \(A\) is of \textit{type-E} (for Equal) if \(p(A) = |\hat{B}| = |B|\); otherwise, \(p(A) = |\hat{B}| < |B|\) and we say that A is of \textit{type-L} (for Less).
\end{definition}

\fixx{We build on Navarro's solution \cite{navarro2019document} for counting on CFGs, which uses an enhanced grid where points count all the occurrences they trigger. The grid ranges are found with the more recent technique~\cite{christiansen2020optimal} that takes $O(m\log n)$ time. Further, we treat type-E rules exactly as Christiansen et al.~\cite{christiansen2020optimal} handle the run-length rules in their specific RLCFGs, as described in Section~\ref{sec:prevcount}. This is possible because type-E rules, by definition, satisfy Property~\ref{prop:basic}. Their method, however, assumes that no two symbols \( B \not= B' \) have the same expansion. To relax this assumption, symbols \( B \) with the same expansion should collectively contribute to the same entries of \( C(\cdot, s) \) and \( C'(\cdot, s) \). We thus index those tables using \( \kappa(\exp(B)) \) rather than \( B \), and for simplicity write \( C(\pi,s)\), \( C'(\pi,s)\), and \( s(\pi) \), where \( \pi = \exp(B) \).}

\fixx{Since each primary occurrence is found in exactly one rule, we can decompose the process of counting by adding up the occurrences found inside type-E and type-L rules. We are then left with the more complicated problem of counting occurrences found from type-L rules. We start with another observation.}

\begin{observation}\label{obs:period_2}
    If  \(A \rightarrow B^s\) is a type-L rule, then \(|B| \geq 2|\hat{B}|\)
\end{observation}
\begin{proof}
    If \(A\) is a type-L rule then \(p(A) = |\hat{B}| < |B| \). In addition, by Lemma \ref{lem:period_div}, \(|\hat{B}|\) divides \(|B|\). Therefore \(|B| \geq 2|\hat{B}|\)
\end{proof}

\fixx{
For type-L rules, we will generalize the strategy of Section~\ref{sec:prevcount}: the cases where $|Q| \le 2|\hat{B}|$ will be handled by adding points to the enhanced grid; in the other cases we will use new data structures that exploit the fact (to be proved) that $Q$ is periodic. Note that each partition $P = R \mid Q$ may correspond to different cases for different run-length rules, so our technique will consider all the cases for each partition.}

\subsection{\fixx{Case $|Q| \leq 2 |\hat{B}|$}} %$ \iff q \geq m - 2|\hat{B}|$}}
\label{sec:type-1}

To capture the primary occurrences with cut $P = R \mid Q$ inside type-L rules $A \to B^s$ where \( |Q| \leq 2|\hat{B}| \), we will incorporate the points \( (x_p, y_p') = (\rev{\exp(\hat{B})}, \exp(\hat{B})) \) and \( (x_p, y_p'') = (\rev{\exp(\hat{B})}, \exp(\hat{B})^2) \) into the enhanced grid outlined in Sections~\ref{sec:grammarindex} and \ref{sec:prevcount}, assigning the values \(-(s-1) \cdot c(A)\) and \( 2 \cdot (s-1) \cdot c(A) \) to each, respectively. %\fixx{Those points will capture all occurrences triggered by a primary occurrence belonging to the sets \(O(A,P,q)\) for \(q \geq m - 2|\hat{B}|\)}. 
The point \( (x_p, y_p') \) will capture the occurrences where \( |R|, |Q| \leq |\hat{B}| \). Note that these occurrences will also find the point \( (x_p, y_p'') \), so the final result will be \( (2 - 1) \cdot (s-1) \cdot c(A) = (s-1) \cdot c(A) \). %In these cases, we only count \( s \) primary occurrences, as there is only one valid alignment.

The point \( (x_p, y_p'') \) will also account for the primary occurrences where \( |R| \leq |\hat{B}| \) and \( |\hat{B}| < |Q| \leq 2|\hat{B}| \). Observation~\ref{obs:period_2} establishes that \( |B| \geq 2|\hat{B}| \), so for each such primary occurrence of cut \( R \mid Q \), with offset \( j \) in \( \exp(A) \), there is a second primary occurrence at \( j - |\hat{B}| \) with cut $P = R' \mid Q'$, where \( |\hat{B}| < |R'|=|R|+|\hat{B}| \leq 2|\hat{B}| \) and \( |Q'| = |Q|-|\hat{B}| \leq |\hat{B}| \). This second cut will not be captured by the points we have inserted because $|R'| > |\hat{B}|$. The other \removed{offsets}\fixx{occurrences} where $P$ matches to the left \fixx{of $j-|\hat{B}|$} fall within \( B \) (and thus are not primary), because we already have $|Q'| \le |\hat{B}|$ in this second occurrence. Thus, for each of the \( s \) copies of $B$ (save the last), we will have two primary occurrences. This yields a total of \( 2 \cdot (s-1) \cdot c(A) \) occurrences, which are properly counted in the points \( (x_p, y_p'') \). See Figure~\ref{fig:g_grammar_tree_primary_secondary_occ}.

\begin{figure}[t]
    \centering
    \scriptsize
    \resizebox{0.7\textwidth}{!}{
        \begin{tikzpicture}
            % Styles
            \tikzset{
                B/.style={fill=lightgray!20, minimum width=2.7cm, minimum height=0.4cm},
                BC/.style={fill=lightgray!20, minimum width=1.3cm, minimum height=0.4cm},
                R/.style={align=right, minimum width=0.9cm, minimum height=0.6cm},
                Q/.style={text=gray, align=left, minimum width=0.9cm, minimum height=0.6cm},
                BP/.style={fill=lightgray, minimum width=1.3cm, minimum height=0.4cm},
                BG/.style={fill=lightgray!20, minimum width=3.7cm}  
            }

            % Nodes
            \node [] (a) at (5.5,5) {\small{$A$}};
            \node [B] (b1) at (2.7,4) {\small{$B$}};
            \node [BP] (bp1_1) at (2,3.4) {\small{$\hat{B}$}};
            \node [BP] (bp1_2) at (3.4,3.4) {\small{$\hat{B}$}};
            \node [B] (b2) at (5.5,4) {\small{$B$}};
            \node [BP] (bp2_1) at (4.8,3.4) {\small{$\hat{B}$}};
            \node [BP] (bp2_2) at (6.2,3.4) {\small{$\hat{B}$}};
            \node [B] (b3) at (8.3,4) {\small{$B$}};
            
            \node [R] (r2) at (3.4,2.9) {\small{$\textcolor{white}{(cgta)^2c}\text{c}$}};
            \node [Q] (q2) at (4.8,2.9) {\small{$(\text{cgta})^2\text{c}$}};
            \node [Q] (q3) at (6.2,2.9) {\small{$\text{c}\textcolor{white}{(cgta)^2c}$}};
            
            \node [R] (r3) at (2,2.6) {\small{$\textcolor{white}{(cgta)^2c}\text{c}$}};
            \node [Q] (q3_1) at (3.4,2.6) {\small{$(\text{cgta})^2\text{c}$}};
            \node [Q] (q3_2) at (4.8,2.6) {\small{$\text{c}\textcolor{white}{(cgta)^2c}$}};

            % Edges
            \draw [gray] (a) to (b1);
            \draw [gray] (a) to (b2);
            \draw [gray] (a) to (b3);

            % Matrix
            \matrix[matrix of nodes, nodes={align=right}] (m) at (-1,4) {
                \small{$\exp(A) = (\text{(cgta)}^{2}\text{c})^6$} \\
                \small{$\exp(B) = (\text{(cgta)}^{2}\text{c})^2$} \\
                \small{$\exp(\hat{B}) = \text{(cgta)}^{2}\text{c}$} \\
                \small{$R = \text{c}$} \\
                \small{\textcolor{gray}{$Q = (\text{cgta})^2\text{cc}$}} \\ 
            };
        \end{tikzpicture}
    }
    \caption{We show the occurrences captured by the point \((x_p, y_p'') = (\exp(\hat{B}), \exp(\hat{B})^2)\). Note how the occurrence in the first row is correctly captured by \((x_p, y_p'')\), whereas that in the second row is not captured by any point. Consequently, the first row is effectively counted twice. Given that the point \((x_p, y_p'')\) is assigned a weight of \(2 \cdot (s-1) \cdot c(A)\), the total number of occurrences is \(4 \cdot c(A)\).}
    \label{fig:g_grammar_tree_primary_secondary_occ}
\end{figure}

\subsection{\fixx{Case $|Q| > 2|\hat{B}|$}}
\label{sec:case2}
\fixx{We first show that, for $Q$ to be longer than $2|\hat{B}|$ in some run-length rule, $P$ must be periodic.}

%\textcolor{blue}{
\begin{lemma}\label{lem:p_Y}
    Let \(P\), with \(p = p(P)\), have a primary occurrence with cut \(P = R \mid Q\) in the rule \( A \to B^s \), with \(p(A) = |\hat{B}|\) and \(|Q| > 2|\hat{B}|\). Then it holds that \( p = p(A)\).
\end{lemma}
%}
\begin{proof}
%\textcolor{blue}{ 
Since \( |P| \ge |\hat{B}| \) and \( P \) is contained within \(\exp(A) = \exp(\hat{B})^{s'}\), by branch 3 of Definition \ref{def:period}, \(|\hat{B}|\) must be a period of \(P\). Thus, \( p = p(P) \leq |\hat{B}| \). Suppose, for contradiction, that \( p < |\hat{B}| \). According to Lemma~\ref{lem:period}, because \fixx{\(|\hat{B}| \le |Q|/2 \le |P|/2\)} is a period of \( P \), it follows that \( p \) divides \(|\hat{B}|\). Since $\exp(\hat{B})$ is contained in $P$, again by branch 3 of Definition \ref{def:period} it follows that $p < |\hat{B}| \le |B|$ is a period of $\exp(B)$, and thus of $\exp(A)$, contradicting the assumption that \( p(A) = |\hat{B}| \). Hence, we conclude that \( p = |\hat{B}| \).
\end{proof}

\fixx{Note that $P$ is then periodic because $p(P) = p(A) = |\hat{B}| < |Q|/2 \le |P|/2$, and $Q$ is also periodic by branch 3 of Def.~\ref{def:period}, because it occurs inside $P$ and $|Q| \ge 2p$.}
\fixx{The following definition will help us show that we capture every primary occurrence exactly once.}
\fixx{
\begin{definition} \label{def:alignment}
% Como excluye el caso de Q corto, creo que es mejor no hablar de todas las occs. 
%Let $Occ$ be the set of all the primary occurrences of a pattern $P$. 
The {\em alignment} of a primary occurrence $x$ found with cut $P = R \mid Q$ inside the type-L rule $A \to B^s$ is $\mathit{align}(x) = 1 + ((|R|-1) \bmod |\hat{B}|)$. 
%Let $Occ_r = \{ x \in Occ : \mathit{offs}(x)=r \}$. 
\end{definition}
}

\fixx{The definition is sound because every primary occurrence is found using exactly one cut $P = R \mid Q$. Note that $\mathit{align} \in [1\dd |\hat{B}|]$ is the %The sets $Occ_r$ partition the primary occurrences according to the 
distance from the starting position of an occurrence, within $\exp(A)$, to the start of the next copy of $\exp(\hat{B})$. We will explore all the possible cuts of $P$, but each rule $A \to B^s$ will be probed only with the cuts where $1 \le |R| \le |\hat{B}|$. From those cuts, all the corresponding primary occurrences aligned with the $s'-1$ boundaries between copies of $\hat{B}$ (i.e., with the same alignment, $|R|$) will be captured.}
%Our approach systematically covers all the primary occurrences by grouping their subsets $Occ_r$ according to $|Q|$.}

\fixx{We distinguish two subcases, depending on whether $Q$ is longer than $B$ or not. If it is, we must ensure that in the alignments we count the occurrence is fully within $\exp(A)$. If it is not, we must ensure that the alignments we count do correspond to primary occurrences (i.e., they cross a border between copies of $B$). }
\subsubsection{\fixx{Case $2|\hat{B}| < |Q| \leq |B|$}} % \iff m - |B| \leq q < m - 2|\hat{B}| $} 
\label{sec:Fpi}
%\color{black}
%No veo que ayude
%Lemma~\ref{lem:p_Y} allows us to classify these occurrences into two scenarios: one where \( p = |\hat{B}| \) (i.e., \( |Q| > 2|\hat{B}| \)), and another where this condition does not hold.

To handle \fixx{this case}\removed{the case for \( |Q| > 2|\hat{B}| \)}, we construct a specific data structure based on the period \( |\hat{B}| \). %If \( Q > |B| \), it would correspond to a type-2 occurrence, which requires a specific analysis. Therefore, we will focus our study on occurrences where \( 2|\hat{B}| < |Q| \le |B| \). 
The proposed solution is supported by the following lemma.

\begin{lemma} \label{lem:period_q_b}
    Let \(P\), with \(p = p(P)\), have a primary occurrence with cut \(P = R \mid Q\) in the \fixx{type-L} rule \( A \to B^s \), with \(p(A) = |\hat{B}|\),  \fixx{\( |R| \le |\hat{B}| \)}, and \( 2|\hat{B}| < |Q| \leq |B| \). Then, the number of primary occurrences of $P$ in \( \exp(A) \) is \( (s-1) \cdot \ceil{|Q|/p}\).
\end{lemma}

\begin{proof}
Since \fixx{\( |R| \le |\hat{B}| \)}, $R$ can be aligned at the end of the \( |B|/|\hat{B}| \) positions where $\exp(\hat{B})$ starts in $\exp(B)$. No other alignments are possible for the cut $R \mid Q$ because, by Lemma~\ref{lem:p_Y}, $p=|\hat{B}|$ and another alignment would imply that $P$ aligns with itself with an offset smaller than $p$, a contradiction by branch 2 of Definition~\ref{def:period}.

Those alignments correspond to primary occurrences only if \fixx{$P$} does not fall completely within $\exp(B)$. The \removed{number of }alignments that correspond to primary occurrences \removed{is}\fixx{are} then \fixx{those where $R$ is aligned at the end of the last} \( \ceil{|Q|/|\hat{B}|} \) \fixx{ending positions of copies of $\hat{B}$}, all of which start within $\exp(B)$ because $|Q| \le |B|$. 
%Therefore, the total number of valid alignments is \( \min\left(|B|/|\hat{B}|, \ceil{|Q|/|\hat{B}|}\right) \) 
%Given that \( Q \leq |B| \), this simplifies to \( \ceil{|Q|/|\hat{B}|} \), which 
This is equivalent to \( \ceil{|Q|/p} \), as \( p = |\hat{B}| \) by Lemma \ref{lem:p_Y}. Thus, the number of primary occurrences of $P$ in \( A \) is \fixx{\( (s-1) \cdot \ceil{|Q|/p}\)}. 
%$ = \sum_{q = m-|B|}^{m-2|\hat{B}|-1}{|O(A,P,q)|}\)}. 
See Figure~\ref{fig:type_b_case_q_leq_2Bp}.
\end{proof}

\begin{figure}[t]
    \centering
    \resizebox{0.7\textwidth}{!}{
        \begin{tikzpicture}

            % Definición de estilos de nodos
            \tikzset{
                B/.style={fill=lightgray!20, minimum width=3.9cm, minimum height=0.4cm},
                BC/.style={fill=lightgray!20, minimum width=0.9cm, minimum height=0.4cm},
                T/.style={align=center, minimum width=0.9cm, minimum height=0.6cm},
                BP/.style={fill=lightgray, minimum width=0.9cm, minimum height=0.4cm},
                R/.style={fill=darkgray, minimum width=0.9cm, minimum height=0.3cm},
                Q/.style={fill=mgray, minimum width=1.8cm, minimum height=0.3cm},
                Qtail/.style={fill=mgray, minimum width=0.9cm, minimum height=0.3cm}
            }

            % Nodos principales
            \node [] (a)  at  (9,5) {\small{$A$}};
            
            % Reglas de tipo B
            \node [B] (brev1) at  (3.5,4) {\small{$B$}};
            \node [BP] (brev2) at  (2,3.4) {\small{$\hat{B}$}};
            \node [BP] (brev3) at  (3,3.4) {\small{$\hat{B}$}};
            \node [BP] (brev4) at  (4,3.4) {\small{$\hat{B}$}};
            \node [BP] (brev5) at  (5,3.4) {\small{$\hat{B}$}};
            \node [B] (brev6) at  (7.5,4) {\small{$B$}};
            \node [BP] (brev7) at  (6,3.4) {\small{$\hat{B}$}};
            \node [BP] (brev8) at  (7,3.4) {\small{$\hat{B}$}};
            \node [BP] (brev9) at  (8,3.4) {\small{$\hat{B}$}};
            \node [BP] (brev10) at  (9,3.4) {\small{$\hat{B}$}};
        
            % Expresiones en los nodos
            \node [] (exp1) at  (5,2.9) {\small{$\text{\textcolor{white}{cgt}a}$}};
            \node [] (exp2) at  (6,2.9) {\small{$\text{\textcolor{gray}{cgta}}$}};
            \node [] (exp3) at  (7,2.9) {\small{$\text{\textcolor{gray}{cgta}}$}};
            \node [] (exp4) at  (8,2.9) {\small{$\text{\textcolor{gray}{c}\textcolor{white}{gta}}$}};
            
            \node [] (exp5) at  (4,2.6) {\small{$\text{\textcolor{white}{cgt}a}$}};
            \node [] (exp6) at  (5,2.6) {\small{$\text{\textcolor{gray}{cgta}}$}};
            \node [] (exp7) at  (6,2.6) {\small{$\text{\textcolor{gray}{cgta}}$}};
            \node [] (exp8) at  (7,2.6) {\small{$\text{\textcolor{gray}{c}\textcolor{white}{gta}}$}};
            
            \node [T] (exp9) at  (3,2.3) {\small{$\text{\textcolor{white}{cgt}a}$}};
            \node [T] (exp10) at  (4,2.3) {\small{$\text{\textcolor{gray}{cgta}}$}};
            \node [T] (exp11) at  (5,2.3) {\small{$\text{\textcolor{gray}{cgta}}$}};
            \node [T] (exp12) at  (6,2.3) {\small{$\text{\textcolor{gray}{c}\textcolor{white}{gta}}$}};

            % Reglas de tipo BC
            \node [BC] (brev11) at  (10,4) {\small{$B$}};
            \node [BC] (brev12) at  (11,4) {\small{$B$}};

            % Matriz para las definiciones de R y Q
            \matrix[matrix of nodes, nodes={align = right}] (m) at (-1,4) {
                \small{$\exp(A) = (\text{cgta})^{16}$} \\
                \small{$\exp(B) = (\text{cgta})^4$} \\
                \small{$\exp(\hat{B}) = \text{cgta}$} \\
                \small{$R = \text{a}$} \\
                \small{\textcolor{gray}{$Q = \text{cgtacgtac}$}} \\
            };
              % Líneas de conexión entre A y los nodos B
            \draw[gray] (a) -- (brev1);
            \draw[gray] (a) -- (brev6);
            \draw[gray] (a) -- (brev11);
            \draw[gray] (a) -- (brev12);
          
        \end{tikzpicture}
    }
    \caption{
    \fixx{If $2|\hat{B}|<|Q|\le|B|$, there are \(\ceil{|Q|/p}\) primary occurrences around the boundary between any two blocks \(B\) (we zoom on one) with the cut \(P = R \mid Q\). We show the possible alignments of $P$ below the blocks \(\hat{B}\). For a rule \(A \rightarrow B^s\) there are \((s-1)\) boundaries, yielding \((s-1) \cdot \ceil{|Q|/p}\) primary occurrences. In this case, \(\ceil{|Q|/p} = 3\) and \(s - 1 = 3\), yielding \(9\) primary occurrences.}
    \removed{The \(\ceil{|Q|/p}\) primary occurrences around a boundary between two blocks \(\hat{B}\) for a partition \(P = R \mid Q\). The \textcolor{blue}{thin} block rows represent the different possible alignments of the pattern, where dark gray denotes \(R\) and light gray corresponds to \(Q\). For a rule \(A \rightarrow B^s\) there are \((s-1)\) boundaries, so the total number of occurrences is \((s-1) \cdot \ceil{|Q|/p}\).}}

    \label{fig:type_b_case_q_leq_2Bp}
\end{figure}

% \begin{lemma} \label{lem:period_b_q_2b}
% Let \( P = R \,|\, Q \)  with \(p = p(P)\) such that \( |R| < |\hat{B}| \), \( Q > 2|\hat{B}| \), and \( |B| < Q \leq 2|B| \). Then, the number of valid primary occurrences in \( A \) is \( \frac{|B|}{p} \).
% \end{lemma}

% \begin{proof}
% Since \( |R| < |\hat{B}| \), the number of possible left shifts is \( \frac{|B|}{|\hat{B}|} \), as the length of \( B \) determines how many blocks of \( \hat{B} \) can fit into \( B \). It must be ensured that the pattern does not fall completely within \( B \); however, this is guaranteed because \( Q > |B| \). Therefore, the total number of valid alignments is \( \frac{|B|}{|\hat{B}|} \). By Lemma \ref{lem:period_q_bp} \( p = |\hat{B}| \) therefore \( \frac{|B|}{p} \) as stated.
% \end{proof}
Based on Lemma \ref{lem:period_q_b} we introduce our first period-based data structure. Considering the solution \removed{to particular cases }described in Section~\ref{sec:prevcount}, \fixx{where Property~\ref{prop:basic} holds,} the challenge with \fixx{type-L} rules $A \rightarrow B^s$ \fixx{(i.e., rules} that differ from their transformed version $A \rightarrow \hat{B}^{s'}$\fixx{)} is that the number of alignments with cut $R \mid Q$ inside $\exp(A)$ is \(s' - \ceil{|Q|/p}\), \fixx{but $B$ does not determine $p=p(A)$}. We will instead use $\hat{B}$ to index those nonterminals $A$.

\no{
\textcolor{olive}{\todo{En color oliva la version con la estructura de grilla que sugirio el revisor 1.}
For each Type-L rule \( A \to B^s \) (\( A \to \hat{B}^{s'} \) being its transformed version), we compute its \emph{signature} \(\kappa(\exp(\hat{B}))\) (refer to Section~\ref{sec:kr}) and store it in a perfect hash table \( H \). Each entry in \( H \), corresponding to a specific signature \(\kappa(\pi)\), is associated with a grid \( F_{\pi} \). For every rule \( A_i \to B_i^{s_i} \) such that \( \kappa(\exp(\hat{B_i})) = \kappa(\pi) \), we add two points to the grid: \((x_p, y_p') = (\rev{\exp(\hat{B})}, \exp(B)) \quad \text{and} \quad (x_p, y_p'') = (\rev{\exp(\hat{B})}, \exp(B)^2)\). These points are assigned weights of \((s-1) \cdot c(A)\) and \(-(s-1) \cdot c(A)\), respectively. The purpose of the negative weights is to exclude occurrences where \(|Q| < 2|\hat{B}|\). Together, these points account for the \((s-1) \cdot c(A)\) terms. Finally, we multiply by \(\ceil{|Q|/p}\) to complete the computation.}
}

For each type-L rule $A \to B^s$ (\( A \rightarrow \hat{B}^{s'} \) being its transformed version), we compute its signature \(\kappa(\exp(\hat{B}))\) (recall Section~\ref{sec:kr}) and store it in a perfect hash table \( H \). Each entry in table \( H \), which corresponds to a specific signature \(\kappa(\pi)\), will be linked to an array \( F_{\pi} \). Each position \( F_{\pi}[i] \) represents a \fixx{type-L} rule \( A_i \to B_i^{s_i} \) where \( \kappa(\exp(\hat{B_i})) = \kappa(\pi) \). The rules \( A_i \) are sorted in $F_\pi$ by \removed{increasing}\fixx{decreasing} lengths \( |B_i| \). We also store a field \fixx{with the cumulative sum}
\begin{equation*}
    F_{\pi}[i].sum = \sum_{1 \le j \leq i} (s_j - 1) \cdot c(A_j).
\end{equation*}

% For the occurrences corresponding to Lemma \ref{lem:period_b_q_2b}, we need to construct an array \( F'_{\pi} \) similar to \( F_{\pi} \), where we define \( F'_{\pi}[i] = \sum_{0 < j \leq i} |B| \times s \times c(A_j) \). To determine the total occurrences, we must identify the range \( (s, e) \) in \( F'_{\pi} \) such that \( A_s \) is the first rule for which \( Q > |B| \), and \( A_e \) is the last rule for which \( Q \leq 2|B| \). Once we have found this range, the total occurrences will be calculated as \( F'_{\pi}[e] - F'_{\pi}[s] \).

%This structure is used as follows.
Given a pattern \(P[1\dd m]\), we first calculate its shortest period \(p = p(P)\). For each cut $P=R \mid Q$ with \fixx{\(1 \le |R| \le \min(p,m-2p-1)\)}, we compute \(\kappa(\pi)\) for \(\pi = Q[1\dd p]\) to identify the corresponding \no{\textcolor{olive}{grid} }array \(F_{\pi}\) in \(H\). Note that we only consider the cuts \(R \mid Q\) where \fixx{\(|R| \le p\)}\fixx{, as this corresponds precisely to $|R| \le |\hat{B}|$ for the rules stored in $F_{\pi}$; note $p=|\pi|$}\removed{to avoid overcounting, as other cuts yield repeated strings \(\pi\)}. In addition, \fixx{the condition $|R| \le m-2p-1$ ensures} that \(|Q| > 2p = 2|\hat{B}|\)\fixx{, so we are correctly enforcing the condition of this subsection and focusing on the occurrences of each alignment $\mathit{align}=|R|$ one by one}.  We will find in $H$ every (transformed) rule $A \rightarrow \hat{B}^{s'}$ where $\hat{B} = \pi$, sharing the period $p$ with $Q$, as well as its prefix $\pi = \exp(B)[1\dd p] = Q[1\dd p]$. Once we have obtained the \no{\textcolor{olive}{grid} }array \fixx{$F_{\pi}$}, \removed{we calculate the range as previously explained and add the total number of occurrences to the overall count.
In order to count the total number of occurrences,}we \removed{first}find the largest $i$ such that \removed{\( |B_i| < |Q| \)}\fixx{$|B_i| \ge |Q|$}\removed{, so that any rule \( A_k \) with \( k > i \) will satisfy \( |Q| \leq |B_k| \), implying that \( B_k \) contains \( Q \)}. The number of \fixx{primary} occurrences \fixx{for the cut $P=R\mid Q$ in type-L rules where $2|\hat{B}| < |Q| \le |B|$} is then \( \removed{(F_{\pi}[|F_{\pi}|].sum - F_{\pi}[i].sum)}\fixx{F_{\pi}[i].sum} \cdot \ceil{|Q|/p} \).
%}

%\begin{figure}[t]
%    \begin{center}
%        \resizebox{0.8\textwidth}{!}{
%            \includegraphics[]{images/type-1-occ.png}
%        }
%    \end{center}
%    \vspace*{-1cm}
%    \caption{A representation of a point \((x,y') = (\rev{\exp(B)}, \exp(B))\). Black blocks represent the period of rule \(A\). The gray blocks on the bottom represent a single type-1 occurrence with its \(\lceil (m - q)/p \rceil\) left-shifts. \textcolor{red}{poner $\hat{B}$ en vez de $A_{1p}$ y $B$ en vez de $A_1$.}
%    }
%    \label{fig:type_1_occ}
%\end{figure}
%
%Figure \ref{fig:type_1_occ} shows an example of a type-1 occurrence. In this case, we count an occurrence in which $P[q+1\dd m]$ spans more than \(2p(\exp(A))\) and does not exceed \(\exp(B)\). Because we do not analyze all cuts in the pattern, we must count this occurrence three times, matching the number of possible left-shifts in the original rule \(B\). Those left-shifts are not counted by any of the remaining cuts, as we only consider the first \(p(P)\) cuts. 
%
%The run-length rules are stored only in this second grid, so together with the first one they amount to $O(g_{rl})$ space.

\subsubsection{\fixx{Case $|Q| >|B|$}} % \iff q < m - |B|$} 
\label{sec:type-2}
\removed{As stated in point 2.b of the general architecture of our solution, we will not distinguish between type-E and type-L rules for type-2 occurrences.}
Our analysis \fixx{for the remaining case} is grounded on the following lemma.

\begin{lemma}    \label{lem:longQ}
%    Let $P$, with \( p = p(P) \), have a type-2 primary occurrence in $A$ with cut $P=R\mid Q$, with \( |R| < p \), and \( |Q| > 2p \). Then it holds that \( p = p(A) \). \textcolor{red}{Fíjate si está bien meter el $|Q|>2p$ como consecuente y no como antecedente: 
Let $P$, with \( p = p(P) \), have a \removed{type-2} primary occurrence in \fixx{a type-L\removed{(E-Type)} rule} $A \rightarrow B^s$ with cut $P=R\mid Q$, with \fixx{\( |R| \le p \)} \fixx{and \(|Q| > |B|\) \removed{\((|Q| > 2|B|)\)}}. Then it holds that \( p = p(A) \) and $|Q| > 2p$. %}
%    \textcolor{blue}{Lo consideré como consecuente porque, al analizar un corte del patrón y realizar la búsqueda en la grilla, debemos asegurarnos de que \(Q > 2p\). De lo contrario, en las ocurrencias de tipo-E podríamos contar ocurrencias de tipo-1, ya que si \(Q \leq 2p\), entonces \( |B| < Q \leq 2|B|\), y la condición para la búsqueda en la grilla es que \(Q > |B|\). Esto es principalmente para poder usar la grilla con ambos tipos de reglas, evitando recurrir a la solución de Christiansen. No obstante, es cierto que, para los efectos del lema, no es necesario como antecedente, ya que la condición de tener una ocurrencia de tipo-2 lo garantiza. OK, LO REEMPLAZO
%}
\end{lemma}
\begin{proof}
If \(A\)\removed{ is a type-E rule and \(P\) has an \removed{type-2} occurrence within \(A\) such that  \(|Q| > 2|B|\), from Lemma~\ref{lem:p_Y}, it follows that \( p = p(A) = |B| \); further, $|Q|>2p$. \fixx{On the other hand}, if \(A\)} is a type-L rule and \(P\) has an \removed{type-2} occurrence within \(A\) such that  \(|Q| > |B|\), then we have \(|Q| > |B| \geq 2|\hat{B}|\) (by Observation~\ref{obs:period_2}). Since we can express \(A\) as \(A \rightarrow \hat{B}^{s'}\), we can similarly use Lemma~\ref{lem:p_Y} to conclude that \( p = p(A) = |\hat{B}| \); further, $|Q|>2p$.
\removed{So in both types of rules it holds that $p=p(A)$ and $|Q| > 2p$.}
\end{proof}

\fixx{Analogously to Lemma~\ref{lem:p_Y}, Lemma~\ref{lem:longQ}} establishes that, when \( Q \) is sufficiently long, it holds that \( p(P) = p(A) \)\removed{irrespectively of the rule type}\fix{, so}\removed{. This ensures that} all pertinent rules of the form \( A \rightarrow B^s \) can be classified according to their minimal period, \( p(A) \). This period coincides with \( p = p(P) \) when \( P \) has 
\removed{type-2}an occurrence in \fixx{a type-L\removed{(type-E)} rule such that \(|Q| > |B|\) \removed{(\(|Q| > 2|B|\))}}. Further, $|Q|>2p$.
%--- Y sacamos esto: note that in a type-2 occurrence, if $A$ is a type-E rule, then it holds that $|Q| > 2|B| = 2p$, and if $A$ is a type-L rule, then it holds that $|Q| > |B| \ge 2|\hat{B}| = 2p$. So in either case it holds that $|Q| > 2p$.}

% In type-2 occurrences, we have demonstrated that, regardless of the type of rule, if \(Q\) is long enough, then \(p = p(\exp(A))\). This means that we capture all relevant rules \(A \rightarrow B^s\) if we organize them according to their shortest period \(p(\exp(A))\), as it must match \(p = p(P)\) whenever \(Q\) contains type-2 occurrences in \(A\). It is important to note that we are only considering in this case the cuts for which \(|Q| > 2p\) and \(|R| < p\), and that for all rules \(A\) (type-E or type-L), it holds that \(|Q| > 2 \cdot p(A)\). 

\fixx{We also need an analogous to Lemma~\ref{lem:period_q_b} for the case $|Q|>|B|$; this is given next.}

\begin{lemma} \label{lem:period_q}
    \fixx{Let \(P\), with \(p = p(P)\), have a primary occurrence with cut \(P = R \mid Q\) in the type-L rule \( A \to B^s \), with \(p(A) = |\hat{B}|\),  \( |R| \le |\hat{B}| \), and \( |Q| > |B| \). Then, the number of primary occurrences of $P$ in \( \exp(A) \) is \( s'-\ceil{|Q|/p}\).}
\end{lemma}

\begin{proof}
\fixx{Since \( |R| \le |\hat{B}| \), $R$ can be aligned at the end of the \( s'\) positions where $\exp(\hat{B})$ starts in $\exp(A)$. By the same argument of the proof of Lemma~\ref{lem:period_q_b}, no other alignments are possible for the cut $R \mid Q$. Unlike in Lemma~\ref{lem:period_q_b}, all those alignments correspond to primary occurrences, because $Q$ is always long enough to exceed $B$. Also unlike in Lemma~\ref{lem:period_q_b}, $Q$ may exceed $A$, in which case the occurrence must not be counted in this rule. The alignments that must not be counted are then those where $R$ is aligned at the end of the last \( \ceil{|Q|/|\hat{B}|} \) ending positions of copies of $\hat{B}$. 
This is equivalent to \( \ceil{|Q|/p} \), as \( p = |\hat{B}| \) by Lemma~\ref{lem:longQ}. Thus, the number of primary occurrences of $P$ in \( A \) is \( s'- \ceil{|Q|/p}\).} \fixx{See Figure \ref{fig:type_b_case_q_g_B}.}
\end{proof}
\begin{figure}[t]
    \centering
    \resizebox{0.7\textwidth}{!}{
        \begin{tikzpicture}

            % Definición de estilos de nodos
            \tikzset{
                B/.style={fill=lightgray!20, minimum width=1.9cm, minimum height=0.4cm},
                BC/.style={fill=lightgray!20, minimum width=0.9cm, minimum height=0.4cm},
                T/.style={align=center, minimum width=0.9cm, minimum height=0.6cm},
                BP/.style={fill=lightgray, minimum width=0.9cm, minimum height=0.4cm},
                BG/.style={fill=lightgray!20, minimum width = 3.7cm}
            }

            % Nodos principales
            \node [] (a)  at  (5.5,5) {\small{$A$}};
            
            % Reglas de tipo B
            \node [B] (brev1) at    (2.5,4) {\small{$B$}};
            \node [BP] (brev2) at   (2,3.4) {\small{$\hat{B}$}};
            \node [BP] (brev3) at   (3,3.4) {\small{$\hat{B}$}};
            
            \node [B] (brev4) at    (4.5,4) {\small{$B$}};
            \node [BP] (brev5) at   (4,3.4) {\small{$\hat{B}$}};
            \node [BP] (brev6) at   (5,3.4) {\small{$\hat{B}$}};
            
            \node [B] (brev7) at    (6.5,4) {\small{$B$}};
            \node [BP] (brev8) at   (6,3.4) {\small{$\hat{B}$}};
            \node [BP] (brev9) at   (7,3.4) {\small{$\hat{B}$}};

            \node [B] (brev10) at    (8.5,4) {\small{$B$}};
            \node [BP] (brev11) at   (8,3.4) {\small{$\hat{B}$}};
            \node [BP] (brev12) at   (9,3.4) {\small{$\hat{B}$}};
            
            % Expresiones en los nodos
            \node [] (exp1) at  (2,2.9) {\small{$\text{\textcolor{white}{cgt}a}$}};
            \node [] (exp2) at  (3,2.9) {\small{$\text{\textcolor{gray}{cgta}}$}};
            \node [] (exp3) at  (4,2.9) {\small{$\text{\textcolor{gray}{cgta}}$}};
            \node [] (exp4) at  (5,2.9) {\small{$\text{\textcolor{gray}{c}\textcolor{white}{gta}}$}};
            
            \node [] (exp5) at  (3,2.6) {\small{$\text{\textcolor{white}{cgt}a}$}};
            \node [] (exp6) at  (4,2.6) {\small{$\text{\textcolor{gray}{cgta}}$}};
            \node [] (exp7) at  (5,2.6) {\small{$\text{\textcolor{gray}{cgta}}$}};
            \node [] (exp8) at  (6,2.6) {\small{$\text{\textcolor{gray}{c}\textcolor{white}{gta}}$}};
            
            \node [] (exp9) at  (4,2.3) {\small{$\text{\textcolor{white}{cgt}a}$}};
            \node [] (exp10) at  (5,2.3) {\small{$\text{\textcolor{gray}{cgta}}$}};
            \node [] (exp11) at  (6,2.3) {\small{$\text{\textcolor{gray}{cgta}}$}};
            \node [] (exp12) at  (7,2.3) {\small{$\text{\textcolor{gray}{c}\textcolor{white}{gta}}$}};
            
            \node [] (exp13) at  (5,2) {\small{$\text{\textcolor{white}{cgt}a}$}};
            \node [] (exp14) at  (6,2) {\small{$\text{\textcolor{gray}{cgta}}$}};
            \node [] (exp15) at  (7,2) {\small{$\text{\textcolor{gray}{cgta}}$}};
            \node [] (exp16) at  (8,2) {\small{$\text{\textcolor{gray}{c}\textcolor{white}{gta}}$}};
            
            \node [] (exp17) at  (6,1.7) {\small{$\text{\textcolor{white}{cgt}a}$}};
            \node [] (exp18) at  (7,1.7) {\small{$\text{\textcolor{gray}{cgta}}$}};
            \node [] (exp19) at  (8,1.7) {\small{$\text{\textcolor{gray}{cgta}}$}};
            \node [] (exp20) at  (9,1.7) {\small{$\text{\textcolor{gray}{c}\textcolor{white}{gta}}$}};

            % Matriz para las definiciones de R y Q
            \matrix[matrix of nodes, nodes={align = right}] (m) at (-1,4) {
                \small{$\exp(A) = \text{(cgta)}^{8}$} \\
                \small{$\exp(B) = \text{(cgta)}^2$} \\
                \small{$\exp(\hat{B}) = \text{cgta}$} \\
                \small{$R = \text{a}$} \\
                \small{\textcolor{gray}{$Q = \text{cgtacgtac}$}} \\
            };

            % Líneas de conexión entre A y los nodos B
            \draw[gray] (a) -- (brev1);
            \draw[gray] (a) -- (brev4);
            \draw[gray] (a) -- (brev7);
            \draw[gray] (a) -- (brev10);
        \end{tikzpicture}
    }
    \caption{If \(|Q| > |B|\), we can compute all occurrences of \(P\) around blocks \(\hat{B}\) without the risk of any occurrence being fully contained in a block \(B\): the number of primary occurrences of \(P\) in \(\exp(A)\) is simply \(s' - \ceil{|Q|/p}\). In this example, with \(s' = 8\) and \(\ceil{|Q|/p} = 3\), there are 5 occurrences.}
    \label{fig:type_b_case_q_g_B}
\end{figure}

% We then compute, for each transformed rule \(A \rightarrow \hat{B}^{s'}\)\textcolor{blue}{(type-E and type-L)}, its signature \(\kappa(\exp(\hat{B}))\) and store it in a perfect hash table \(H\)
We then enhance table \(H\), introduced in Section~\ref{sec:Fpi}, with a second period-based data structure. Each entry in table \(H\), corresponding to some $\kappa(\pi)$, will additionally store a grid \(G_\pi\). In this grid, each row represents a \fixx{type-L} rule \(A \rightarrow B^s\) whose transformed version is $A \rightarrow \hat{B}^{s'}$, that is, such that \(\pi = \exp(\hat{B}) = \exp(B)[1\dd p]\). The rows are sorted by increasing lengths \(|B|\) (note $|B|\ge|\pi|=p$ for all $B$ in $G_\pi$). The columns represent the different exponents \(s'\) of the transformed rules. \fixx{The row of rule $A \to B^s$ has then a unique point at column $s'$,}\removed{Each point in the grid represents a \fixx{type-L} rule \(A \fixx{\to B^s}\),} and we associate two values with it: \(C'_\pi(A) = c(A)\) and \(C''_\pi(A) = s' \cdot c(A)\). Since  no rule appears in more than one grid, the total space for all grids is in \(O(g_{rl})\). 

% \begin{figure}[t]
%     \begin{center}
%         \resizebox{0.6\textwidth}{!}{
%             \includegraphics[]{images/grilla.png}
%         }
%     \end{center}
%     \vspace*{-1cm}
%     \caption{A representation of a grid \(G_\pi\). Black lines define the grid range $[x,n] \times [1,y]$. Each point in the range represents a rule \(A \rightarrow \hat{B}^{s'}\) where a type-2 occurrence appears.\textcolor{red}{Cambia $B_y$ por $B$.}}
%     \label{fig:grilla}
% \end{figure}

\begin{figure}[t!]
    \begin{minipage}{0.40\textwidth} \scriptsize
        \begin{itemize} 
            \item \(S \rightarrow X_1 X_2 \texttt{t} X_7 X_8 X_9 X_{11} \)
            \item \(X_1 \rightarrow \texttt{cgta} \)
            \item \(X_2 \rightarrow X_{1}^4 \)
            \item \(X_3 \rightarrow \texttt{cg}  \)
            \item \(X_4 \rightarrow \texttt{ta} \)
            \item \(X_5 \rightarrow X_{3} X_{4} \)
            \item \(X_6 \rightarrow X_{1} X_{5} \)
            \item \(X_7 \rightarrow X_{6}^4 \)
            \item \(X_8 \rightarrow X_{3}^4 \)
            \item \(X_9 \rightarrow X_{12}^5 \)
            \item \(X_{10} \rightarrow X_{2} X_{5} \)
            \item \(X_{11} \rightarrow X_{10}^4 \)
            \item \(X_{12} \rightarrow X_{1} X_{1} \texttt{cca} \)
        \end{itemize}

        \tiny
       \begin{tabular}{r|cr }
           % \(s\) & \cellcolor{mgray!30}\(4\)  \\ \hline
            \(C(\mathtt{cgta},4)\) & \cellcolor{mgray!30}\(5\)  \\ \cline{1-2}
          \(C'(\mathtt{cgta},4)\) & \cellcolor{mgray!30}\(20\) & type-E hash \\ \cline{1-2}
          \(s(\mathtt{cgta})\) & \cellcolor{mgray!30}\(\{ 4 \}\)
        \end{tabular} 

    \end{minipage}
    \hspace{-3cm}
    \begin{minipage}{0.70\textwidth}
        \resizebox{\textwidth}{!}{
      
        \begin{tikzpicture}
            \tikzset{TypeB/.style={fill=mgray, draw}}
                
    % X7
    \node [] (13)    at (6  ,  4) {\small{$X_6$}};
    \node [] (15)    at (5  ,  3) {\small{$X_1$}};
    \node [] (16)    at (7  ,  3) {\small{$X_5$}};
    \node [] (17)    at (6.5  ,  2) {\small{$X_3$}};
    \node [] (19)    at (5.8  ,  1) {\small{\texttt{c}}};
    \node [] (20)    at (6.5  ,  1) {\small{\texttt{g}}};
    
    \draw [] (17) to (19);
    \draw [] (17) to (20); 
    
    \node [] (18)    at (7.5  ,  2) {\small{$X_4$}};
    \node [] (21)    at (7.2  ,  1) {\small{\texttt{t}}};
    \node [] (22)    at (7.8  ,  1) {\small{\texttt{a}}};
    
    \draw [] (18) to (21);
    \draw [] (18) to (22);

    \draw [] (16) to (17);
    \draw [] (16) to (18);
    
    \draw [] (13) to (15);
    \draw [] (13) to (16);
    % X7^3
    \node [] (14)    at (8  ,  4) {\small{$X_{6}^{[3]}$}};
    % X1
    % X2
    \node [] (2)    at (2.1   ,  5) {\small{$X_1$}};
    \node [] (8)    at (1   ,  4) {\small{\texttt{c}}};
    \node [] (10)   at (3.2   ,  4) {\small{\texttt{a}}};
    \node [] (11)   at (1.75  ,  4) {\small{\texttt{g}}};
    \node [] (12)   at (2.5  ,  4) {\small{\texttt{t}}};
    \draw [] (2) to (11);
    \draw [] (2) to (12);
    \draw [] (2) to (8);
    \draw [] (2) to (10);
    %X3
    \node [draw, rectangle] (31)  at (4.5  ,  5) {\small{$X_2$}};
    \node [] (32)  at (4  ,  4) {\small{$X_1$}};
    \node [] (33)  at (4.8  ,  4) {\small{$X_{1}^{[3]}$}};
    \draw [] (31) to (32);
    \draw [] (31) to (33);
    %T
    \node [] (3)  at (6  ,  5) {\small{\texttt{t}}};
    % X8
    \node [TypeB] (4)  at (7.25  ,  5) {\small{$X_7$}};
    \draw [] (4) to (13);
    \draw [] (4) to (14);
    % X9
    \node [draw, rectangle] (5)  at (9.5  ,  5) {\small{$X_8$}};
    \node [] (23)  at (9  ,  4) {\small{$X_3$}};
    \node [] (24)  at (10  ,  4) {\small{$X_{3}^{[3]}$}};
    \draw [] (5) to (23);
    \draw [] (5) to (24);
    
    % X10
    \node [draw, rectangle] (6)  at (11.5  ,  5) {\small{$X_{9}$}};
    \node [] (25)  at (11.5  ,  3.95) {\small{$X_{12}$}};
    \node [] (26)  at (12.5  ,  4) {\small{$X_{12}^{[4]}$}};
    \draw [] (6) to (25);
    \draw [] (6) to (26);
    %X12
    \node [] (34)  at (10.5  ,  3) {\small{$X_{1}$}};
    \node [] (35)  at (11  ,  3) {\small{$X_{1}$}};
    \node [] (36)  at (11.5  ,  3) {\small{\texttt{c}}};
    \node [] (37)  at (12  ,  3) {\small{\texttt{c}}};
    \node [] (38)  at (12.5  ,  3) {\small{\texttt{a}}};
    \draw [] (25) to (34);
    \draw [] (25) to (35);
    \draw [] (25) to (36);
    \draw [] (25) to (37);
    \draw [] (25) to (38);
    % X11
    \node [TypeB] (7)  at (14  ,  5) {\small{$X_{11}$}};
    \node [] (27)  at (13.5  ,  4) {\small{$X_{10}$}};
    \node [] (28)  at (14.5  ,  4) {\small{$X_{10}^{[3]}$}};
    \draw [] (7) to (27);
    \draw [] (7) to (28);
    \node [] (29)  at (13.2  ,  3) {\small{$X_{2}$}};
    \node [] (30)  at (13.8  ,  3) {\small{$X_{5}$}};
    \draw [] (27) to (29);
    \draw [] (27) to (30);
    
    % S
    \node [] (1)  at (8  ,	7) {\small{$S$}};
    \draw [] (1) to (2);
    \draw [] (1) to (3);
    \draw [] (1) to (4);
    \draw [] (1) to (5);
    \draw [] (1) to (6);
    \draw [] (1) to (7);
    \draw [] (1) to (31);
        
        \end{tikzpicture}
    }
   
    \end{minipage}
    
    \vspace{-16.0mm}
     \begin{flushright}
            \tiny
             \begin{tabular}{r | c c }
                                                    &   \cellcolor{mgray!30}8                   & \cellcolor{mgray!30}20 \\ \hline
               % \cellcolor{mgray!30} \(X_1^{(4)}\)   &   \cellcolor{mgray!50} \(X_2^{(5,20)}\)   &   \cellcolor{mgray!50}                    & \cellcolor{mgray!50} \\% \hline
               \cellcolor{mgray!30} \(X_6^{(8)}\)   &   \cellcolor{mgray!50} \(X_7^{(1,8)}\)    &  \cellcolor{mgray!50}        \\% \hline
                            \(X_{10}^{(20)}\)       &                                           & \(X_{11}^{(1,20)}\)               \\ \hline
            \end{tabular} \\ \ \\
    
            \makebox[\textwidth][r]{
                \makebox[5cm][c]{
                    $\small G_{\texttt{cgta}}$
                }
            }
    \end{flushright}

    \vspace{-2mm}
    \begin{minipage}{0.8\textwidth} \tiny
         \begin{tabular}{l | c@{} c@{} c@{~~} c@{~~~} c@{~~~} c@{~~~}}
                                                            & \(\texttt{cg}\)   & \((\texttt{cg})^2\)   & \(\texttt{cgta}\)     & \((\texttt{cgta})^2\)   &  \cellcolor{mgray!30}\(\texttt{(cgta)}^2\texttt{cca}\)    &  \cellcolor{mgray!30}\((\texttt{(cgta)}^2\texttt{cca})^2\)    \\ \hline 
             \rowcolor{mgray!20}\(\texttt{acc(atgc)}^2\)    &                   &                &                   &                         &                              \cellcolor{mgray!50}\(X_{9}^{(1)}\) &    \cellcolor{mgray!50}\(X_{9}^{(3)}\)    \\% \hline
             \rowcolor{mgray!20}\(\texttt{atgc}\)           &                   &                       & \(X_2^{(5)}\)\(X_7^{(-3)}\)\(X_{11}^{(-3)\ }\)           & \(X_2^{(10)}\)\(X_7^{(6)}\)\(X_{11}^{(6)}\) &  \cellcolor{mgray!50}&  \cellcolor{mgray!50} \\% \hline
             % \rowcolor{mgray!20}\((\texttt{atgc})^2\)    &                   &                   &                   &                       &                       &\\% \hline
             % \rowcolor{mgray!20}\((\texttt{atgc})^5\)    &                   &                   &                   &                       &                       & \\% \hline
             % \(\texttt{c}\)                              & \(X_9^{(1)}\)       & \(X_9^{(3)}\)       &                   &                       &                       &  \\% \hline
             \(\texttt{gc}\)                               & \(X_8^{(1)}\)        &\(X_8^{(2)}\)            &                       &  &\cellcolor{mgray!30} & \cellcolor{mgray!30}\\ \hline
         \end{tabular}
        \centering
        
    \end{minipage}  
    \begin{minipage}{0.1\textwidth} \tiny \centering
        \begin{tabular}{|c|c| }\hline
            \(X_7\) & \cellcolor{mgray!50}\(X_{11}\)  \\ \hline
          \(3\) & \cellcolor{mgray!50}\(6\)  \\ \hline
        \end{tabular}
        \vspace{-2.0mm}
        \begin{center}
            $\!\!\!\!\!\!\!F_\texttt{cgta}$    
         \end{center}
    \end{minipage}

    \caption{On top, a RLCFG on the left and its grammar tree on the right. \fixx{Type-E rules are enclosed in white rectangles and Type-L rules in gray rectangles. Below the rules we show the values \(C(B,s)\) and \(C'(B,s)\) \cite{christiansen2020optimal} we use to handle the E-type rules (see Section \ref{sec:prevcount}); we only show those for $\exp(X_1)=\mathtt{cgta}$.} On the bottom left we show the points we add to the standard grid\removed{in order to handle the type-1 occurrences of rules $A \rightarrow B^s$}. The points for type-E rules are represented as \( A^{(c(A))} \) and \( A^{((s-2) \cdot c(A))} \) and those for type-L rules as \( A^{(-(s-1) \cdot c(A))} \) and \( A^{(2 \cdot (s-1) \cdot c(A))} \). The bottom right shows the grid $G_{\pi}$ and the array $F_{\pi}$ for the transformed rules $A \rightarrow \hat{B}^{s'}$ where $\hat{B} = \pi = \texttt{cgta}$. In $F_\pi$ we show the fields $F[i].sum$. In $G_\pi$, the row labels show $B^{(|B|)}$ and the column labels show $s'$; the points show $A^{(C', C'')}$. \fixx{Consider the cut $P=\texttt{a} \mid \texttt{cgtacgtac}$, with $p(P)=4$. We identify \(9\) occurrences in type-E rules: \(4\) are found within the rule \(X_9\) using the standard grid, while the remaining \(5\) are determined via the values of \(C(X_1, s)\) and \(C'(X_1, s)\). These \(5\) occurrences specifically arise within \(\exp(X_2) = (\texttt{cgta})^4\). Similarly, in the type-L rules, we detect \(14\) occurrences: \(9\) occur within the rule \(X_{11}\), identified using the \(F_{\texttt{cgta}}\) array, and the remaining \(5\) arise within \(\exp(X_7) = (\texttt{cgta})^8\), captured using the \(G_{\texttt{cgta}}\) grid.} The final two occurrences of this cut are located using standard CFG rules at $\exp(S)[4\dd 13]$ ($X_1 \cdot X_2$) and $\exp(S)[111\dd 120]$ ($X_9 \cdot X_{11}$).}

    \label{fig:ejemplo}    
\end{figure}
% \textcolor{red}{($1$ inside $\exp(X_2)=(\texttt{cgta})^4$, where $X_2$ appears $5$ times, and $5$ inside $\exp(X_7)=(\texttt{cgta})^8$, which appears once). The other $5$ occurrences of $P$ are found with normal CFG rules: one at $\exp(S)[4\dd 13]$ ($X_1 \cdot X_2$), one at $\exp(X_{10})[8\dd 17]$ ($X_2 \cdot X_5$, $q=6$), and 3 copies of the latter in $\exp(X_{11})$}

Given a pattern \(P[1\dd m]\), we proceed analogously as explained at the end of Section~\ref{sec:Fpi} in order to identify $F_\pi$: We compute \(p = p(P)\), and for each cut $P=R\mid Q$ with \fixx{\(1 \le |R| \le \min(p,m-2p-1)\)}, we calculate \(\kappa(\pi)\), for $\pi=Q[1\dd p]$, to find the corresponding grid \(G_\pi\) in $H$. \fixx{On the type-L rules $A\to B^s$, this tries out every possible alignment $\mathit{align} = |R|$, one by one, from $1$ to $|\hat{B}|$. The limit $|R| < m-2p$ can also be set because, by Lemma~\ref{lem:longQ}, it must hold $|Q| > 2|\hat{B}|$ on the rules of $G_{\pi}$ we find with the cut $P=R \mid Q$.}
% Note we consider only the splits $P[1\dd q] \cdot Q$ where $q<p$ to avoid overcounting; the others yield repeated strings $\pi$. We will find in $H$ every (transformed) rule $A \rightarrow \hat{B}^{s'}$ where $\hat{B} = \pi$, sharing the period $p$ with $Q$, and its prefix: $\pi = \exp(B)[1\dd p] = Q[1\dd p]$. 

\fixx{We must enforce two conditions on the rules of $G_{\pi}$ to consider: (a) $|Q| > |B|$ as corresponds in this subsection, and (b) $s'-\lceil|Q|/p\rceil \ge 0$, that is, $Q$ fits within $\exp(A)$. The complying rules then contribute}
\removed{\fixx{When} $|Q| > |B|$, it holds $\exp(B) \sqsubset Q$. Those are precisely the \removed{type-2} occurrences to count, where such rules contribute }$c(A) \cdot (s'-\lceil|Q|/p\rceil) = C''_\pi(A)-C'_\pi(A) \cdot \lceil|Q|/p\rceil$ \fixx{by Lemma~\ref{lem:period_q}.}\removed{, provided \fixx{$s'-\lceil|Q|/p\rceil \ge 0$}, that is, $Q$ fits within $\exp(A)$.\removed{Note that, since we have a match with $\pi = Q[1 \dd p]$ and \(|Q| > 2p\), we have that for type-E rules, \(|Q| > 2|B|\). Consequently, \(Q\) does not lie within any $B$ in a type-E rule in \(G_\pi\), and we only need to ensure that it fits within \(A\). Therefore for type-E rules we will count occurrences \removed{of type-2} \fixx{\(|Q| > 2|B|\)} without need to verify that $|Q| > 2|B|$.}} 
%We also limit $|R| < m-2p$, to ensure $|Q| > 2p$. \gonzalo{Por qué tenemos que asegurar eso? No vale por el lema 10 cuando las otras condiciones se cumplen? Es correcto decir que lo podemos limitar porque sabemos que debe valer?}
%\textcolor{red}{No puede pasar, si $|B| < |Q| \le 2|B|$, que la contemos como type-2+type-L y como type-1+type-E?}
%\textcolor{blue}{Si puede pasar, pero serian reglas diferentes. En el ejemplo annadi la regla de tipo-A  $X_9 \rightarrow  (cgtacgtacca)^5$ para $q = 1$, encontrariamos $4$ occurencias de tipo-1.}

To enforce those conditions, we find in $G_\pi$ the largest row $y$ representing a rule $A \rightarrow B^s$ such that \(|B| < |Q|\). We also find the smallest column \(x\) where \fixx{\((s' =)\, x \ge \lceil |Q|/p \rceil\)}. \removed{It is easy to see that the}\fixx{The} points in the range $[x,n] \times [1,y]$ of the grid \fixx{then} correspond to the set $D$ of \fixx{type-L} run-length rules where we have a\removed{type-2} \fixx{primary} occurrence \fixx{with $|Q|>|B|$.}\removed{: each point satisfies \(|Q| > |B|\) (recall that this check is needed only for type-L rules)and \fixx{\(s'-\lceil |Q|/p \rceil \geq 0\)} (i.e., it fits within $\exp(A)$). \fixx{Those primary occurrences belong to the sets \(O(A,P,q)\) for \(q < m - |B|\) (or \(q < m - 2|B|\) in type-E rule).}} We aggregate the values \(C'_\pi\) and \(C''_\pi\) from the range, \fixx{which} yields the correct sum of all\removed{type-2} \fixx{the pertinent} occurrences:
\begin{equation*}
     \left(\sum_{A\rightarrow B^s \in D} C''_{\pi}(A)\right) - \left(\sum_{A\rightarrow B^s \in D} C''_{\pi}(A)\right) \cdot \left\lceil|Q|/p\right\rceil 
    ~~=~~ \sum_{A\rightarrow B^s \in D} {c(A) \cdot s' - c(A)\cdot \left\lceil|Q|/p\right\rceil. }
\end{equation*}

 Figure~\ref{fig:ejemplo} gives a \fixx{thorough} example.

\subsection{The final result}

\removed{For type-1 occurrences} \fixx{Our structure} extends the grid \fixx{of Section~\ref{sec:prevcount},} built for non-run-length rules, with one point per run-length rule\fixx{: those of type-E are handled as described in Section~\ref{sec:prevcount} and those of type-L as in Section~\ref{sec:oursol}. Thus} the structure is of size $O(g_{rl})$ and range queries on the grid take time $O(\log^{2+\epsilon} g_{rl})$. Occurrences on such a grid are counted in time $O(m\log n + m\log^{2+\epsilon} g_{rl})$ \cite[Thm.~A.5]{christiansen2020optimal}. \fixx{This is also the time to count the occurrences in type-E rules for our solution, and those in type-L rules when $|Q| \le 2|B_p|$ (Section~\ref{sec:type-1}).}

For \fixx{our period-based data structures} \removed{type-2 occurrences} (Section\fixx{s ~\ref{sec:Fpi} and } \ref{sec:type-2}), we calculate $p(P)$ in \(O(m)\) time \cite{crochemore2002jewels}, and compute all prefix signatures of $P$ in $O(m)$ time as well, so that later any substring signature is computed in $O(1)$ time (Section~\ref{sec:kr}). The limits in the arrays $F_\pi$ and in the grids $G_\pi$ can be found with exponential search in time $O(\log m)$ (we might need to group rows/columns with identical values to achieve this \fixx{time}). The range sums for $C'_\pi$ and $C''_\pi$ take time $O(\log^{2+\epsilon} g_{rl})$. They are repeated for each of the $O(m)$ cuts of $P$, adding up to time $O(m\log^{2+\epsilon} g_{rl})$. \fixx{Those are then within the previous time complexities as well.} \removed{\gonzalo{Ojo: Type-2} occurrences then add $O(g_{rl})$ space and $O(m\log^{2+\epsilon} g_{rl})$ time to the general scheme.} 

\begin{theorem}
Let a RLCFG of size $g_{rl}$ represent a text \(T[1\dd n]\). Then, \fixx{for
any constant $\epsilon>0$,} we can build in $O(n\log n)$ expected time an index of size \(O(g_{rl})\) that counts the number of occurrences of a pattern \(P[1\dd m]\) in \(T\) in time \(O(m \log n + m\log^{2+\epsilon} g_{rl}) \subseteq O(m\log^{2+\epsilon} n)\).
\end{theorem}

Just as for previous schemes \cite{christiansen2020optimal}, the construction time is dominated by the $O(n\log n)$ expected time to build the collision-free Karp--Rabin functions \cite{bille2014time}.

\shorten{
Note that the bulk of the search cost are the geometric queries, which are easily done in $O(\log n)$ 
time if we store cumulative sums in all the levels of the data structure \cite{Cha88,navarro2019document}. 
%This raises the space to $O(g_{rl}\log n)$ and reduces the total counting time to $O(m\log n)$. 
\fixx{More generally, setting Navarro's $\epsilon$ to $1/\log^{1-\delta} g_{rl}$ \cite[Thm.~3]{navarro2019document}, we obtain the following tradeoff.}

\begin{corollary} \label{col:tradeoff}
\fixx{Let a RLCFG of size $g_{rl}$ represent a text \(T[1\dd n]\). Then, for any constant $0 \le \delta < 1$, we can build in $O(n\log n)$ expected time an index of size \(O(g_{rl} \log^{1-\delta} g_{rl})\) that counts the occurrences of a pattern \(P[1\dd m]\) in \(T\) in time \(O(m \log n + m\log^{1+\delta} g_{rl}) \subseteq O(m\log^{1+\delta} n)\).}
\end{corollary}
} 

\shorten{
\subsection{An application}

Recent work \cite{Gao22,Navcpm23.1} shows how to compute the maximal exact matches (MEMs) of $P[1\dd m]$ in $T[1\dd n]$, which are the maximal substrings of $P$ that occur in $T$, in case $T$ is represented with an arbitrary RLCFG. Navarro \cite{navarro2023computing} extends the results to $k$-MEMs, which are maximal substrings of $P$ that occur at least $k$ times in $T$. To obtain good time complexities for large enough $k$, he resorts to counting occurrences of substrings $P[i\dd j]$ with the grammar. His Thm.~7, however, works only for CFGs, as no efficient counting algorithm existed on RLCFGs. In turn, his Thm.~8 works only for a particular RLCFG. We can now state his result on an arbitrary RLCFG; by his Thm.~11 this also extends to ``$k$-rare MEMs''.

\begin{corollary}[{cf.~\cite[Thm.~7]{navarro2023computing}}]
Let a RLCFG of size $g_{rl}$ generate only $T[1\dd n]$. Then, for
any constant $\epsilon>0$, we can build a data structure of size $O(g_{rl})$ that finds the $k$-MEMs of any given pattern $P[1\dd m]$, for any $k>0$ given with $P$, in time $O(m^2 \log^{2+\epsilon} g_{rl})$.
\end{corollary}
}

\shorten{
\section{Conclusion}
We \removed{closed} \fixx{have presented the first solution to} the problem of counting the occurrences of a pattern in a text represented by an arbitrary RLCFG, which was posed by Christiansen et al.~\cite{christiansen2020optimal} in 2020 and solved only for particular cases.
 This required combining solutions to CFGs \cite{navarro2019document} and particular RLCFGs \cite{christiansen2020optimal}, but also new insights for the general case. The particular existing solutions required that $|B|$ is the \fix{shortest} period of $\exp(A)$ in rules $A \rightarrow B^s$. While this does not hold in general RLCFGs, we proved that, except in some borderline cases that can be handled separately, the shortest periods of the pattern and of \fix{$\exp(A)$} must coincide. While the particular solutions could associate $\exp(B)$ with the period of the pattern, we must associate many strings \fix{$\exp(A)$} that share the same shortest period, and require a more sophisticated geometric data structure to collect only those that qualify for our search. Despite those complications, however, we manage to define a data structure of size \(O(g_{rl})\) from a RLCFG of size $g_{rl}$, that counts the occurrences of $P[1\dd m]$ in $T[1\dd n]$ in time \(O(m\log^{2+\epsilon} n)\) for any constant  \(\epsilon > 0\), the same result that existed for the simpler case of CFGs. Our approach extends the applicability of arbitrary RLCFGs to cases where only CFGs could be used, \fixx{equalizing} the available tools to handle both types of grammars\removed{ at the same level}.
}

\bibliographystyle{splncs04}

\bibliography{paper}

\begin{thebibliography}{10}
\providecommand{\url}[1]{\texttt{#1}}
\providecommand{\urlprefix}{URL }
\providecommand{\doi}[1]{https://doi.org/#1}

\bibitem{BEGV18}
Bille, P., Ettienne, M.B., G{\o}rtz, I.L., Vildh{\o}j, H.W.: Time-space trade-offs for {Lempel-Ziv} compressed indexing. Theoretical Computer Science  \textbf{713},  66--77 (2018)

\bibitem{BLRSRW15}
Bille, P., Landau, G.M., Raman, R., Sadakane, K., Rao, S.S., Weimann, O.: Random access to grammar-compressed strings and trees. SIAM Journal on Computing  \textbf{44}(3),  513--539 (2015)

\bibitem{bille2014time}
Bille, P., G{\o}rtz, I.L., Sach, B., Vildh{\o}j, H.W.: Time--space trade-offs for longest common extensions. Journal of Discrete Algorithms  \textbf{25},  42--50 (2014)

\bibitem{CLLPPSS05}
Charikar, M., Lehman, E., Liu, D., Panigrahy, R., Prabhakaran, M., Sahai, A., Shelat, A.: The smallest grammar problem. IEEE Transactions on Information Theory  \textbf{51}(7),  2554--2576 (2005)

\bibitem{Cha88}
Chazelle, B.: A functional approach to data structures and its use in multidimensional searching. SIAM Journal on Computing  \textbf{17}(3),  427--462 (1988)

\bibitem{christiansen2020optimal}
Christiansen, A.R., Ettienne, M.B., Kociumaka, T., Navarro, G., Prezza, N.: Optimal-time dictionary-compressed indexes. ACM Transactions on Algorithms (TALG)  \textbf{17}(1),  1--39 (2020)

\bibitem{CN10}
Claude, F., Navarro, G.: Self-indexed grammar-based compression. Fundamenta Informaticae  \textbf{111}(3),  313--337 (2010)

\bibitem{claude2012improved}
Claude, F., Navarro, G.: Improved grammar-based compressed indexes. In: Proc. 19th International Symposium on String Processing and Information Retrieval (SPIRE). pp. 180--192 (2012)

\bibitem{DBLP:journals/jcss/ClaudeNP21}
Claude, F., Navarro, G., Pacheco, A.: Grammar-compressed indexes with logarithmic search time. Journal of Computer and System Sciences  \textbf{118},  53--74 (2021)

\bibitem{crochemore2002jewels}
Crochemore, M., Rytter, W.: Jewels of stringology: text algorithms. World Scientific (2002)

\bibitem{FGHP13}
Ferrada, H., Gagie, T., Hirvola, T., Puglisi, S.J.: Hybrid indexes for repetitive datasets. Philosophical Transactions of the Royal Society A  \textbf{372}(2016),  article 20130137 (2014)

\bibitem{FKP18}
Ferrada, H., Kempa, D., Puglisi, S.J.: Hybrid indexing revisited. In: Proc. 20th Workshop on Algorithm Engineering and Experiments (ALENEX). pp.~1--8 (2018)

\bibitem{periodicity}
Fine, N.J., Wilf, H.S.: Uniqueness theorems for periodic functions. Proceedings of the American Mathematical Society  \textbf{16}(1),  109--114 (1965)

\bibitem{GGKNP12}
Gagie, T., Gawrychowski, P., K{\"{a}}rkk{\"{a}}inen, J., Nekrich, Y., Puglisi, S.J.: A faster grammar-based self-index. In: Proc. 6th International Conference on Language and Automata Theory and Applications (LATA). pp. 240--251. LNCS 7183 (2012)

\bibitem{GGKNP14}
Gagie, T., Gawrychowski, P., K{\"a}rkk{\"a}inen, J., Nekrich, Y., Puglisi, S.J.: {LZ77}-based self-indexing with faster pattern matching. In: Proc. 11th Latin American Symposium on Theoretical Informatics (LATIN). pp. 731--742 (2014)

\bibitem{GNPjacm19}
Gagie, T., Navarro, G., Prezza, N.: Fully-functional suffix trees and optimal text searching in {BWT}-runs bounded space. Journal of the ACM  \textbf{67}(1),  article 2 (2020)

\bibitem{GJL21}
Ganardi, M., Jez, A., Lohrey, M.: Balancing straight-line programs. Journal of the ACM  \textbf{68}(4),  27:1--27:40 (2021)

\bibitem{Gao22}
Gao, Y.: Computing matching statistics on repetitive texts. In: Proc. 32nd Data Compression Conference (DCC). pp. 73--82 (2022)

\bibitem{GKKL18}
Gawrychowski, P., Karczmarz, A., Kociumaka, T., Lacki, J., Sankowski, P.: Optimal dynamic strings. In: Proc. 29th Annual {ACM-SIAM} Symposium on Discrete Algorithms (SODA). pp. 1509--1528 (2018)

\bibitem{Jez15}
Jez, A.: Approximation of grammar-based compression via recompression. Theoretical Computer Science  \textbf{592},  115--134 (2015)

\bibitem{Jez16}
Jez, A.: A really simple approximation of smallest grammar. Theoretical Computer Science  \textbf{616},  141--150 (2016)

\bibitem{karkkainen1996lempel}
K{\"a}rkk{\"a}inen, J., Ukkonen, E.: Lempel-{Z}iv parsing and sublinear-size index structures for string matching. In: Proc. 3rd South American Workshop on String Processing (WSP). pp. 141--155 (1996)

\bibitem{KR87}
Karp, R.M., Rabin, M.O.: Efficient randomized pattern-matching algorithms. IBM Journal of Research and Development  \textbf{2},  249--260 (1987)

\bibitem{KP18}
Kempa, D., Prezza, N.: At the roots of dictionary compression: String attractors. In: Proc. 50th Annual ACM Symposium on the Theory of Computing (STOC). pp. 827--840 (2018)

\bibitem{KK23}
Kempa, D., Kociumaka, T.: Collapsing the hierarchy of compressed data structures: Suffix arrays in optimal compressed space. In: Proc. 64th {IEEE} Annual Symposium on Foundations of Computer Science (FOCS). pp. 1877--1886 (2023)

\bibitem{KY00}
Kieffer, J.C., Yang, E.H.: Grammar-based codes: {A} new class of universal lossless source codes. IEEE Transactions on Information Theory  \textbf{46}(3),  737--754 (2000)

\bibitem{kociumaka2024near}
Kociumaka, T., Navarro, G., Olivares, F.: Near-optimal search time in $\delta$-optimal space, and vice versa. Algorithmica  \textbf{86}(4),  1031--1056 (2024)

\bibitem{KNP23}
Kociumaka, T., Navarro, G., Prezza, N.: Toward a definitive compressibility measure for repetitive sequences. IEEE Transactions on Information Theory  \textbf{69}(4),  2074--2092 (2023)

\bibitem{KRRW15}
Kociumaka, T., Radoszewski, J., Rytter, W., Walen, T.: Internal pattern matching queries in a text and applications. In: Proc. 26th Annual {ACM-SIAM} Symposium on Discrete Algorithms (SODA). pp. 532--551 (2015)

\bibitem{KN13}
Kreft, S., Navarro, G.: On compressing and indexing repetitive sequences. Theoretical Computer Science  \textbf{483},  115--133 (2013)

\bibitem{LM00}
Larsson, J., Moffat, A.: Off-line dictionary-based compression. Proceedings of the IEEE  \textbf{88}(11),  1722--1732 (2000)

\bibitem{LZ76}
Lempel, A., Ziv, J.: On the complexity of finite sequences. IEEE Transactions on Information Theory  \textbf{22}(1),  75--81 (1976)

\bibitem{MST12}
Maruyama, S., Sakamoto, H., Takeda, M.: An online algorithm for lightweight grammar-based compression. Algorithms  \textbf{5}(2),  214--235 (2012)

\bibitem{NavACMcs14}
Navarro, G.: Spaces, trees and colors: The algorithmic landscape of document retrieval on sequences. ACM Computing Surveys  \textbf{46}(4),  article 52 (2014), 47 pages

\bibitem{Navacmcs20.3}
Navarro, G.: Indexing highly repetitive string collections, part {I}: Repetitiveness measures. ACM Computing Surveys  \textbf{54}(2),  article 29 (2021)

\bibitem{Navacmcs20.2}
Navarro, G.: Indexing highly repetitive string collections, part {II}: Compressed indexes. ACM Computing Surveys  \textbf{54}(2),  article 26 (2021)

\bibitem{Navcpm23.1}
Navarro, G.: Computing {MEM}s on repetitive text collections. In: Proc. 34th Annual Symposium on Combinatorial Pattern Matching (CPM). p. article 22 (2023)

\bibitem{NOUspire22}
Navarro, G., Olivares, F., Urbina, C.: Balancing run-length straight-line programs. In: Proc. 29th International Symposium on String Processing and Information Retrieval (SPIRE). pp. 117--131 (2022)

\bibitem{NP18}
Navarro, G., Prezza, N.: Universal compressed text indexing. Theoretical Computer Science  \textbf{762},  41--50 (2019)

\bibitem{navarro2019document}
Navarro, G.: Document listing on repetitive collections with guaranteed performance. Theoretical Computer Science  \textbf{772},  58--72 (2019)

\bibitem{navarro2023computing}
Navarro, G.: Computing {MEMs} and relatives on repetitive text collections. ACM Transactions on Algorithms  \textbf{21}(1),  article 12 (2025)

\bibitem{NMWM04}
Nevill-Manning, C., Witten, I., Maulsby, D.: Compression by induction of hierarchical grammars. In: Proc. 4th Data Compression Conference (DCC). pp. 244--253 (1994)

\bibitem{NIIBT16}
Nishimoto, T., I, T., Inenaga, S., Bannai, H., Takeda, M.: Fully dynamic data structure for {LCE} queries in compressed space. In: Proc. 41st International Symposium on Mathematical Foundations of Computer Science (MFCS). pp. 72:1--72:15 (2016)

\bibitem{raskhodnikova2013sublinear}
Raskhodnikova, S., Ron, D., Rubinfeld, R., Smith, A.: Sublinear algorithms for approximating string compressibility. Algorithmica  \textbf{65},  685--709 (2013)

\bibitem{Ryt03}
Rytter, W.: Application of {L}empel-{Z}iv factorization to the approximation of grammar-based compression. Theoretical Computer Science  \textbf{302}(1-3),  211--222 (2003)

\bibitem{Sak05}
Sakamoto, H.: A fully linear-time approximation algorithm for grammar-based compression. Journal of Discrete Algorithms  \textbf{3}(2–4),  416--430 (2005)

\bibitem{SS82}
Storer, J.A., Szymanski, T.G.: Data compression via textual substitution. Journal of the {ACM}  \textbf{29}(4),  928--951 (1982)

\bibitem{TKNIBT20}
Tsuruta, K., K{\"o}ppl, D., Nakashima, Y., Inenaga, S., Bannai, H., Takeda, M.: Grammar-compressed self-index with {L}yndon words. CoRR  \textbf{2004.05309} (2020)

\end{thebibliography}

\end{document}